\documentclass{IEEEojcsys}
\usepackage[utf8]{inputenc}
\usepackage[colorlinks,urlcolor=blue,linkcolor=blue,citecolor=blue]{hyperref}

\usepackage{amsmath,amsthm,amssymb}
\usepackage{graphicx,setspace}
\usepackage{color}
\usepackage[english]{babel}
\usepackage{mathrsfs,dsfont}
\usepackage{array}
\usepackage{tikz}
\usepackage{adjustbox}
\usepackage{multirow}
\usepackage{float}
\usepackage{booktabs}
\usepackage{graphicx}
\usepackage{lineno}
\usepackage{rotating}
\usepackage{pdflscape}
\usepackage{dsfont}
\usepackage{diagbox}
\usepackage{tikz}
\usepackage{setspace}
\usepackage{mathrsfs}
\usepackage[mathscr]{euscript}
\usepackage{bbm}
\usepackage{mathtools}
\usepackage{subcaption}
\usepackage[linesnumbered,vlined,ruled]{algorithm2e}
\usepackage{accents}

\makeatletter
\@ifundefined{theorem}{
\theoremstyle{plain}
\newtheorem{theorem}{\bf Theorem}
}{}
\@ifundefined{lemma}{
\theoremstyle{plain}
\newtheorem{lemma}[theorem]{\bf Lemma}
}{}
\@ifundefined{corollary}{
\theoremstyle{plain}

}{}
\@ifundefined{claim}{
\theoremstyle{plain}

}{}
\@ifundefined{conjecture}{
\theoremstyle{plain}

}{}
\@ifundefined{prop}{
\theoremstyle{plain}
\newtheorem{prop}[theorem]{\bf Proposition}
}{}
\@ifundefined{example}{
\theoremstyle{definition}

\newtheorem{assumption}{Assumption}
\newtheorem{definition}{Definition}
}{}
\@ifundefined{assumption}{
\theoremstyle{definition}
\newtheorem{assumption}{Assumption}
}{}
\@ifundefined{definition}{
\theoremstyle{definition}
\newtheorem{definition}{Definition}
}{}
\@ifundefined{remark}{
\theoremstyle{plain}
\newtheorem{remark}{Remark}
}{}
\makeatother

\newcommand{\nn}{\nonumber}

\def\QED{~\rule[-1pt]{5pt}{5pt}\par\medskip}
\renewenvironment{proof}{{\bf Proof: \ }}{\hfill \QED}

\let\ALP  \mathcal

\renewcommand{\sp}{\texttt{span}}

\DeclareUnicodeCharacter{03BB}{$\lambda$}

\usepackage{cuted}
\usepackage{algpseudocode}
\usepackage[numbers]{natbib}
\usepackage{enumerate}
\usepackage{algorithm2e}
\usepackage{subcaption}
\usepackage{xcolor}
\usepackage{multicol}
\usepackage{orcidlink}

\newcommand{\mmd}{\textnormal{\texttt{\text{MMD}}}}

\renewcommand{\cal}{\mathcal}

\newcommand{\bb}{\mathbb}
\newcommand{\bd}{\boldsymbol}
\newcommand{\tbb}[1]{\tilde{\mathbb #1}}
\newcommand{\bcal}[1]{\bar{\mathcal #1}}

\newcommand{\msf}{\mathsf}
\newcommand{\mbf}{\mathbf}
\newcommand{\md}{\textnormal{\texttt{ADD}}}
\newcommand{\mtbfa}{\textnormal{\texttt{ARL}}}
\newcommand{\tv}{\textnormal{\texttt{TV}}}

\renewcommand{\sp}{\textnormal{\texttt{span}}}

\newcommand{\calpx}{\cal P(\cal X)}

\newcommand{\rkhs}[1]{\cal H(\msf #1)}
\newcommand{\iprod}[2]{\langle #1 \rangle_{\cal H(\msf {#2})}}

\newcommand{\supnorm}[1][\cdot]{\|#1\|_\infty}

\jvol{00}
\jnum{XX}
\paper{1234567}
\pubyear{2024}
\receiveddate{10 January 2024}
\reviseddate{27 March 2024}
\accepteddate{21 April 2021}
\currentdate{21 April 2024}
\begin{document}


\title{Model-Free Change Point Detection for Mixing Processes} 

\author{HAO CHEN\affilmark{1}\orcidlink{0000-0003-2811-6518} (Student Member, IEEE)}

\author{ABHISHEK GUPTA\affilmark{1}\orcidlink{0000-0003-1117-325X} (Member, IEEE)}

\author{YIN SUN\affilmark{2}\orcidlink{0000-0001-6811-984X} (Member, IEEE)}

\author{NESS SHROFF\affilmark{1}\orcidlink{0000-0002-4606-6879} (Fellow, IEEE)}

\editor{Sonia Martinez}

\affil{Department of Electrical and Computer Engineering, The Ohio State University, Columbus, OH 43210 USA} 
\affil{Department of Electrical and Computer Engineering, Auburn University Auburn, AL 36849, USA} 

\corresp{CORRESPONDING AUTHOR: HAO CHEN (e-mail: \href{mailto:chen.6945@osu.edu}{chen.6945@osu.edu})}

\authornote{This work was supported by Cisco Systems and the U.S. National Science Foundation}

\markboth{Model-Free Change Point Detection for Mixing Processes}{H. CHEN {\itshape ET AL}.}

\begin{abstract}%
    This paper considers the change point detection problem under dependent samples. In particular, we provide performance guarantees for the MMD-CUSUM test under exponentially $\alpha$, $\beta$, and fast $\phi$-mixing processes, which significantly expands its utility beyond the i.i.d. and Markovian cases used in previous studies. We obtain lower bounds for average-run-length ($\mtbfa$) and upper bounds for average-detection-delay ($\md$) in terms of the threshold parameter. We show that the MMD-CUSUM test enjoys the same level of performance as the i.i.d. case under fast $\phi$-mixing processes. The MMD-CUSUM test also achieves strong performance under exponentially $\alpha$/$\beta$-mixing processes, which are significantly more relaxed than existing results. The MMD-CUSUM test statistic adapts to different settings without modifications, rendering it a completely data-driven, dependence-agnostic change point detection scheme. Numerical simulations are provided at the end to evaluate our findings.
\end{abstract}

\begin{IEEEkeywords}
Change point detection, kernel method, mixing processes
\end{IEEEkeywords}

\maketitle

\section{INTRODUCTION}

{Change point detection studies the problem of monitoring for abrupt changes in the statistical properties of an observation sequence, which has been widely considered in the literature \cite{lai1995sequential, dayanik2009sequential, tartakovsky2014sequential, xie2021sequential}. Change point detection has a diverse application that spans many areas, including cybersecurity, network intrusion detection, automated fault monitoring, factory quality control, etc. In many of these application scenarios, one may face various challenges, such as complex unknown dynamics, noisy non-i.i.d observations, and unknown pre- and post-change distributions. Ideally, a completely data-driven method with very few distributional assumptions (independence, density functions, etc.) would be preferred. The goal of this paper is to study the change point detection problem under a completely data-driven setting. To tackle this problem, we employ the MMD-CUSUM statistic proposed in \cite{chen2022change} and analyze its performance under three common mixing conditions, namely $\alpha$, $\beta$, and $\phi$-mixing.}

{The MMD-CUSUM statistic is an extension of the well-known CUSUM statistic \cite{page1954continuous} with the maximum mean discrepancy (MMD). MMD has wide adoption in statistical two-sample tests \cite{gretton2012kernel} and the training of generative adversarial networks \cite{li2017mmd}. As a probability distance, MMD can be easily estimated from samples on general domains (continuous or discrete) without the need for a density function. Thus, it is well suited for change point detection under the completely data-driven setting where pre- and post-change distributions can be unknown.} Additionally, kernel methods have wide compatibility \cite{gartner2003survey, muandet2016kernel} due to the diversity of kernel functions with different data structures, such as discrete data, continuous data, graphical data, etc. Thus, the kernel base method has vast application potential in designing completely data-driven change point detection schemes. In particular, the sequential testing procedures using the maximum mean discrepancy (MMD) have sparked some research interests lately \cite{li2015m, li2019scan, chang2018kernel, flynn2019change, chen2022change}. Most of the existing studies focus on studying the properties of the MMD-based procedures under the i.i.d. case. For continuous state space Markov chains, the MMD-CUSUM test is proposed in \cite{chen2022change} for uniformly ergodic Markov chains, which is known to be hard to satisfy in practice.

Thus, more relaxed assumptions need to be considered to meet the demands of the completely data-driven setting. The main challenge in generalizing the performance analysis of MMD-CUSUM lies in the dependence of samples. Our proposal assumes the mixing property of the stochastic processes generated by the dynamic system. Mixing measures the dependence in the process by its definition \cite{bradley2005basic}, and it is widely considered in extending various results in probability theory to dependent time series \cite{kolmogorov1960strong, adler1961ergodic, ibragimov1962some}. Thus, establishing the performance bounds under various mixing conditions is a natural choice. Furthermore, the mixing conditions we assume highlight the fundamental limit for MMD-CUSUM to achieve a good performance; that is, the speed and strength of the mixing condition the processes satisfy. 

In the current paper, we analyze the performance of MMD-CUSUM under three common mixing conditions, namely $\alpha$, $\beta$, and $\phi$-mixing. We provide bounds on the average-run-length ($\mtbfa$) and average-detection-delay ($\md$) which are the common performance metrics \cite{xie2013sequential}. $\mtbfa$ characterizes how frequently the false alarm occurs and $\md$ characterizes the quickness of the reaction. As outlined in \cite{lai1998information}, the information-theoretic lower bounds are $O(\exp(b))$ for $\mtbfa$ and $O(b)$ for $\md$ for large $b>0$, where $b$ is the threshold parameter. We show that under the fast $\phi$-mixing condition, the MMD-CUSUM achieves these lower bounds and thus is order optimal. Under the exponential $\alpha/\beta$-mixing, $\md$ is bounded by $O(b)$ where $\mtbfa$ is bounded by $O(\exp{b^{\gamma/(\gamma+1)}})$, where $\gamma > 0$ controls the mixing speed (more details in \ref{sec: main results}).

The rest of the paper is organized as follows. Section \ref{sec: background} introduces the necessary background about reproducing kernel Hilbert space and mixing processes. Section \ref{sec: problem} states the problem setting for online change point detection and introduces the MMD-CUSUM test statistic. Section \ref{sec: main results} establishes the main results of this paper. Section \ref{sec: simulation} presents the experiments of the MMD-CUSUM test on synthetic datasets. We conclude the paper with discussions of the limitations and future work in Section \ref{sec: discussion} and \ref{sec: conclusion}.

\subsection{Related works}

Continuous efforts have been made to adapt the kernel two-sample test to a sequential setting, i.e., change point detection. Early work has been focused on detection change in a stream of i.i.d. samples \cite{li2015m, li2019scan, chang2018kernel, flynn2019change}. In \cite{li2015m, li2019scan}, the authors developed a Shewhart chart-type \cite{shewhart1986statistical} procedure that maintains a running estimate of the MMD between a set of curated reference data and incoming samples within a fixed sliding window. Analysis shows strong performance guarantees with an $O(\exp(b^2))$ average-run-length ($\mtbfa$) and an $O(b)$ average detection delay ($\md$), where $b$ is the threshold. However, testing schemes with sliding windows suffer from loss of information as older samples are discarded. To maintain history information, kernel-based CUSUM-type statistics were proposed in \cite{flynn2019change} with an $O(\exp(b))$ average-run-length ($\mtbfa$) and an $O(b)$ average detection delay ($\md$). In \cite{chang2018kernel}, the authors devised a neural network-based kernel selection strategy that finds a kernel whose MMD can best separate the nominal distribution from an adversarial one. The testing scheme is to estimate the MMD with the selected kernel on two adjacent sliding windows. Empirical studying shows promising performance, albeit without theoretical guarantees. 

The analysis of the above methods is based heavily on the i.i.d. assumption. Their technique and results do not carry over naturally to the non-i.i.d. case. Due to the ubiquity of time series data in machine learning, signal processing, economics, and dynamic systems, the i.i.d. assumption limits the application of these methods. More recently, researchers have been adapting the kernel-based change point detection to dependent data. In \cite{chen2022change}, the MMD-CUSUM test is proposed and analyzed under the setting of uniformly ergodic Markov chains on general state space. Recently, \cite{zhang2023data} extended the analysis of MMD-CUSUM to noisy observations of uniformly ergodic Markov chains, i.e., hidden Markov models (HMM). Both cases are special cases of $\phi$-mixing processes \cite{bradley2005basic}. In fact, we show that the same performance can be obtained even when the Markovian and HMM structures are ignored. In other words, the Markov chain and HMM assumptions are not necessary for the performance of the MMD-CUSUM test. Our work even extends to the $\alpha$/$\beta$-mixing processes, which have never been considered for the MMD-CUSUM test previously.  

More broadly, our study falls under the umbrella of the quickest change detection (QCD) theory \cite{veeravalli2014quickest}.  Studies on the QCD problem can be split into two categories: the Bayesian and minimax formulation, depending on the assumption of the change point. The Bayesian formulation, pioneered by \cite{shiryaev1963optimum, shiryaev2007optimal}, places a prior on the distribution of the change point (usually a geometric distribution). Whereas the minimax formulation, first considered by \cite{lorden1971procedures}, assumes the change point is unknown and deterministic. Under both formulations, the different notions of detection delay are minimized while constrained on the probability of false alarm or the false alarm rate ($1/\mtbfa$). A well-known Bayesian QCD formulation is Shiryaev's problem \cite{shiryaev1963optimum}, which seeks the stopping rule that minimizes the average detection delay (under the change point prior) while constrained on the probability of false alarm. The minimax formulations include Lorden's problem \cite{lorden1971procedures} and Pollak's problem \cite{pollak2009optimality}, where the former minimizes the worst-case average delay and the latter minimizes the conditional average delay while both contained on the false alarm rate.

Although the CUSUM statistic was first proposed as a heuristic for the minimax formulation under i.i.d. setting by \cite{page1954continuous}, strong optimality properties have been shown for CUSUM statistic under various settings. Under the i.i.d. setting, exact optimality was shown by \cite{moustakides1986optimal, ritov1990decision} for Lorden's problem. For general non-i.i.d. settings, \cite{lai1998information} has shown that an extension of the CUSUM statistic achieves the information-theoretic lower bound on the conditional average delay (as well as the worst case delay) asymptotically as the false alarm rate goes to 0.

However, the optimality result mentioned previously requires specific knowledge of the pre-and post-change distributions. Furthermore, the QCD problems are intractable for general stochastic processes due to the lack of problem structure \cite{lai1998information}. Thus, the numerous studies on QCD for non-i.i.d settings \cite{lai1995sequential, lai1998information, fuh2003sprt, tartakovsky2005general, dayanik2009sequential, xie2013sequential, tartakovsky2014sequential, fuh2015quickest, tartakovsky2017asymptotic, fuh2018asymptotic, moustakides2019detecting, ford2023exactly} cannot be easily converted to the completely data-driven setting.

\subsection{Contributions}

As a non-parametric model-free change point detection procedure, the MMD-CUSUM test exhibits great potential in completely data-driven applications where distributional assumptions may be difficult to verify. Our performance guarantees under general mixing conditions establish its robustness under dependent samples and further strengthen its capability as a model-free testing scheme. The mixing conditions considered in this paper not only subsume the i.i.d., Markov chain, and HMM settings but also greatly expand beyond those appearing in previous studies on the performance of the MMD-CUSUM test. Our results indicate that the Markovian or HMM structures are not necessary for the strong performance of the MMD-CUSUM test. Additionally, we provide the first performance guarantee for the MMD-CUSUM test under exponentially $\alpha$/$\beta$ and fast $\phi$-mixing processes. Note that stationary exponentially $\beta$-mixing processes include the geometrically ergodic Markov chains as a special case, which violates Doeblin's condition \cite[page 402]{meyn2012markov}. In stark contrast, Doeblin's condition is the core assumption for the performance analysis of the MMD-CUSUM test in \cite{chen2022change} and \cite{zhang2023data}.

\section{BACKGROUND}
\label{sec: background}

In this section, we introduce the necessary background for our discussion. Section \ref{sec: rkhs mmd} collects the usual facts about reproducing kernel Hilbert space (RKHS) and maximum mean discrepancy (MMD). Section \ref{sec: mixing} presents the two notions of mixing used to obtain the main results. Our standard reference is \cite{steinwart2008support} for RKHS and \cite{bradley2005basic} for mixing processes.

\subsection{RKHS and MMD}
\label{sec: rkhs mmd}

Let $(\msf X, \cal X, \bb P)$ be a measure space with Borel $\sigma$-algebra $\cal X$ and $\sigma$-finite measure $\bb P$. Let $\cal P(\cal X)$ denote the set of probability measures over the $\sigma$-algebra $\cal X$. The supremum norm of $f$ is written as $\supnorm[f]\coloneqq\sup_{x\in\msf X} |f(x)|$ and its span is written as $\sp(f)\coloneqq \sup_{x, x'\in \msf X} |f(x) - f(x')|$.

A \textit{reproducing kernel Hilbert space} (RKHS) $\cal H(\msf X)$ on $\msf X$ with kernel $k:\msf X\times\msf X \to \bb R$ is a Hilbert space of real-valued functions on $\msf X$ equipped with inner product $\iprod{\cdot, \cdot}{X}$. The corresponding Hilbert space norm $\|f\|_{\cal H(\msf X)}^2 = \|\|$ The kernel function $k$ satisfies the reproducing property:
\begin{align}
    k(x, \cdot) \in\rkhs{X}\ \text{and } \iprod{f(\cdot), k(x, \cdot)}{X} = f(x), \quad \text{for } x\in\msf X.\nn
\end{align}

The current paper relies on a particular application of RKHS --- Hilbert space embeddings of probability measure. The Hilbert space embedding of $\mu$ under $k$ is written as
\begin{align}
    \cal U(\mu)(\cdot) \coloneqq \int_{\msf X} k(x, \cdot) d\mu,\nn
\end{align}
where $\cal U(\mu)$ is also called the kernel mean embedding of $\mu$. Suppose $\nu\in\calpx$ is another probability measure. One can define a distance function between $\mu$ and $\nu$ using the Hilbert space metric between $\cal U(\mu)$ and $\cal U(\nu)$ 
\begin{align}
    \mmd_k(\mu, \nu) = \|\cal U(\mu) - \cal U(\nu)\|_{\cal H(\msf X)},\nn
\end{align}
which is known as the \textit{maximum mean discrepancy} (MMD) \cite{gretton2012kernel}. {The kernel $k$ such that $\mmd_k(\mu, \nu) = 0 \Leftrightarrow \mu = \nu$ for all $\mu, \nu\in\cal P(\cal X)$ is call a \textit{characteristic kernel} \cite{sriperumbudur2010hilbert}. $\mmd_k$ with a characteristic kernel $k$ is a metric on $\cal P(\cal X)$.}

MMD enjoys a computational advantage, compared with other probability distance functions, such as KL divergence \cite{kullback1951information} and total variation metric (Definition \ref{def: tv metric}), that allows it to be easily estimated empirically for distributions on general domains \cite{gartner2003survey, muandet2016kernel}. 

Let $X_i \sim \mu$ and $X_j'\sim \nu$ for $i=1,\cdots,m$ and $j=1,\cdots,n$. Define their empirical measures as $\hat\mu_m,\hat\nu_n$, respectively. The consistent estimation of the squared MMD is 
\begin{align}
    &\mmd_k^2(\hat\mu_m, \hat\nu_n) = \frac{1}{m^2} \sum_{1\leq i,j\leq N} k(X_i, X_j)\nn\\
    &+ \frac{1}{n^2}\sum_{1\leq i,j\leq M} k(X_i', X_j') - \frac{2}{mn}\sum_{i,j}k(X_i, X_j'),\nn
\end{align}

This was first used by \cite{gretton2012kernel} to propose the kernel two-sample test, and it is the core component of the MMD-CUSUM test studied in the current paper. 

Throughout the paper, we assume the kernel $k$ is \textit{real-valued, measurable, characteristic, and bounded}, {i.e., $\sup_{x\in\msf X} k(x, x) = \bar k < \infty$.} {The boundedness ensures $\mmd_k$ is well-defined.}

\subsection{Mixing processes} 
\label{sec: mixing}

The definitions of the mixing process require the following necessary notations. Consider the space of $\msf X$-valued doubly infinite stochastic processes as $(\msf X^\infty, \cal X^\infty, \tbb P)$ where the indices of a process $\bd X = \{X_i\}_{i\in\bb Z} \in \msf X^\infty$ are allowed to be $-\infty$ and $\infty$. For each index $t\in\bb Z$, let $\cal X_{t}^{\infty}$ denote the $\sigma$-algebra generated by $\{X_i\}_{i=t}^{\infty}$ and $\cal X_{-\infty}^i$ is the $\sigma$-algebra generated by $\{X_t\}_{t=-\infty}^i$. We use $\bcal X_\tau^{\tau'}$ to denote the $\sigma$-algebra generated by $\{X_i\}_{i=-\infty}^{\tau} \cup \{X_i\}_{i=\tau'+1}^\infty$. The marginal probability measure on $\{X_i\}_{i=t}^{\infty}$ is written as $\tbb P_t^\infty$ and the joint probability measure on $\{X_i\}_{i=-\infty}^{t}$ as $\tbb P_{-\infty}^t$. With these notations, we have the definitions of $\alpha$, $\beta$, and $\phi$-mixing coefficients following \cite{doukhan2012mixing, bradley2005basic}.

\begin{definition}[$\alpha$-mixing coefficient]
\label{def: alpha mixing}
    The \textit{$\alpha$-mixing coefficient} \cite{rosenblatt1956central} of a stationary process $\bd X$ is defined as 
    \begin{align}
        \alpha(n) = \sup_t\sup_{A\in\cal X_{-\infty}^t, B\in\cal X_{t+n}^{\infty}} |\tbb P(A\cap B) - \tbb P(A)\cdot\tbb P(B)|.\nn
    \end{align}
    $\bd X$ is called \textit{$\alpha$-mixing} if $\alpha(n)\to 0$ as $n\to\infty$.
\end{definition}
The following $\beta$-mixing coefficient provides a stronger notion of decaying dependence. It can be shown that $2\alpha(n) \leq \beta(n)$ \cite{doukhan2012mixing}. 
\begin{definition}[$\beta$-mixing coefficient]
\label{def: beta mixing}
    The \textit{$\beta$-mixing coefficient} \cite{kolmogorov1960strong} of a stochastic process $\bd X$ is defined as 
    \begin{align}
        \beta(n) = \sup_{t} \sup_{C\in\bcal X_{t+1}^{t+n-1}} |\tbb P(C) - \tbb P_{-\infty}^{t}\otimes\tbb P_{t+1}^\infty(C)|.\nn
    \end{align}
    $\bd X$ is called \textit{$\beta$-mixing} if $\beta(n)\to 0$ as $n\to\infty$. The $\beta$-mixing coefficient can be equivalently written as 
    \begin{align}
        \beta(n) = \sup_t \bb E_{\tbb P_{-\infty}^t} \bigg[\sup_{C\in\cal X_{t+n}^\infty} | \tbb P_{t+n}^\infty(C| \cal X_{-\infty}^{t}) - \tbb P_{t+n}^\infty(C)|\bigg].\nn
    \end{align}
\end{definition} 
Comparing the second definition of $\beta$-mixing with the following definition of $\phi$-mixing, we can see that $\beta(n)\leq\phi(n)$.
\begin{definition}[$\phi$-mixing coefficient]
\label{def: phi mixing}
    The \textit{$\phi$-mixing coefficient} \cite{ibragimov1962some} of a stochastic process $\bd X$ is defined as 
    \begin{align}
        \phi(n) = \sup_t \sup_{B\in\cal X_{-\infty}^{t}} \bigg[\sup_{C\in\cal X_{t+n}^\infty} | \tbb P_{t+n}^\infty(C| B) - \tbb P_{t+n}^\infty(C)|\bigg].\nn
    \end{align}
    $\bd X$ is called \textit{$\phi$-mixing} if $\phi(n)\to 0$ as $n\to\infty$.
\end{definition}

{We say $\bd X$ is \textit{stationary} with respect to $\mu\in\cal P(\cal X)$ if the one-dimensional marginal probability of $X_i$ equals $\mu$ for $\forall i\in\bb Z$. For stationary processes, the supremum over $t$ in the above definitions can be ignored, and one can set $t=0$ without loss of generality. To maintain the simplicity of the presentation, we focus on \textit{stationary} stochastic processes with $\alpha$, $\beta$, and $\phi$-mixing properties in the sequel. However, the results put forward in the current paper can be extended to asymptotically stationary processes, which is discussed in Section \ref{sec: discussion}.}

{The decay rates of the mixing coefficients play an important role in our discussion. The following definitions introduce the \textit{exponential $\alpha/\beta$-mixing} condition and \textit{fast $\phi$-mixing}, which are used throughout the paper.}
\begin{definition}[exponential $\alpha/\beta$-mixing]
\label{def: exp mixing}
    $\bd X$ is said to be exponential $\alpha$ or $\beta$-mixing, if the $\alpha$ or $\beta$-mixing coefficient satisfies
    \begin{align}
        &\alpha(n) \leq \bar\alpha\exp(-cn^\gamma), \quad n\geq 1, \nn\\
        \quad\text{or}\quad
        &\beta(n) \leq \bar\beta\exp(-cn^\gamma), \quad n\geq 1,\nn
    \end{align}
    for $\bar \alpha, \bar \beta,\gamma, c > 0$. 
\end{definition}
\begin{definition}
\label{def: fast phi mixing}[fast $\phi$-mixing]
    $\bd X$ is said to be fast $\phi$-mixing, if the $\phi$-mixing coefficient satisfies
    \begin{align}
        \Phi \coloneqq \sum_{n=0}^\infty \phi(n) < \infty.\nn
    \end{align}
\end{definition}
An exponentially decaying $\phi$-mixing coefficient is certainly summable and thus is covered under the above definition. Definition \ref{def: exp mixing} and \ref{def: fast phi mixing} form the basic assumption on the mixing processes studied in the current paper.

{To bridge the notions of mixing with RKHS, it is convenient to consider the following kernel mixing coefficient introduced in \cite{cherief2022finite}.}
\begin{definition}[kernel mixing coefficient]
\label{def: kernel mixing coefficient}
Let $\bd X$ be a stationary process with distribution $\mu$. For $n\in\bb N$, define the kernel mixing coefficient as
\begin{align}
\label{eqn: kernel mixing coefficient}
    &\rho_k(n) \nn\\
    &= \bigg | \bb E \langle k(X_n, \cdot) - \bb E_\mu k(X, \cdot),  k(X_0, \cdot) - \bb E_\mu k(X, \cdot)\rangle_{\cal H}\bigg |.
\end{align}
We denote the cumulative sum of the kernel mixing coefficient as $\Sigma_\mu \coloneqq \sum_{n=0}^\infty \rho_k(n)$.
\end{definition}
{
    If we treat $\{k(X_i, \cdot)\}_{i\in\bb Z}$ as a sequence of Hilbert space valued stochastic process, then as shown by \cite[Lemma 2.2]{dehling1982almost} $\rho_k(n)$ can bounded by a constant multiple of the $\alpha$-mixing coefficient, i.e., $\rho_k(n) \leq 10\alpha(n)\bar k^2$. Thus, we get $\Sigma_\mu < \infty$ under the assumptions of exponential $\alpha$-mixing,  exponential $\beta$-mixing, and fast $\phi$-mixing.
}

\subsection{Examples of mixing processes}

One notable example of $\phi$-mixing processes is the uniformly ergodic Markov chain. A Markov chain is said to be uniformly ergodic if it is aperiodic and satisfies Doeblin's condition \cite{meyn2012markov}. Thus, it is also called the Doeblin chain. A $q$-th order autoregressive (AR) process is $\phi$-mixing if the Markov chain generated by stacking $q$ consecutive states is a Doeblin chain. The $\phi$-mixing coefficient decays exponentially for uniformly ergodic Markov chains, therefore satisfying the fast $\phi$-mixing condition in Definition \ref{def: fast phi mixing}.

Examples of exponential $\beta$-mixing processes include $V$-geometrically ergodic Markov chains. The Markov transition kernel $P:\msf X\times\cal X \to [0, 1]$ with stationary distribution $\pi$ is said to be $V$-geometrically ergodic if it satisfies
\begin{align}
\label{eqn: geo ergodic mc}
    \tv(P^n(x, \cdot), \pi) \leq V(x)\rho^{\lfloor n/m \rfloor},\quad \text{for all } n,
\end{align}
where $V:\msf X\to[1, \infty)$ is a measurable function, $m$ is an constant integer, and $\rho\in[0, 1)$. When $V$ is bounded on $\msf X$, it becomes the uniform ergodicity condition. From a dynamic system perspective, $V$-geometrically ergodic Markov chains subsume stable nonlinear systems with finite variance additive noise \cite[see][Section 3.5]{vidyasagar2002learning}. The aforementioned examples all work as examples of exponential $\alpha$-mixing processes. Additionally, measurable functionals of $\alpha$, $\beta$, and $\phi$-mixing processes are also  $\alpha$, $\beta$, and $\phi$-mixing processes. The mixing coefficients are bounded by those of the original processes \cite[Lemma 3.6]{vidyasagar2002learning}.

\section{PROBLEM FORMULATION}
\label{sec: problem}

{In this section, we first introduce the online change point detection problem and the commonly used performance metrics \cite[see][]{xie2013sequential, xie2021sequential}. Later, we discuss the proposed MMD-CUSUM test and its properties.}

In the sequel, we make the following assumption and restrict our attention to stochastic processes satisfying the exponential $\alpha/\beta$-mixing and fast $\phi$-mixing conditions in Definition \ref{def: exp mixing}, \ref{def: fast phi mixing}.

\begin{assumption}
    The stochastic processes considered in what follows satisfy one of the three mixing conditions in Definition \ref{def: exp mixing} and \ref{def: fast phi mixing}.
\end{assumption}

\subsection{Online change point detection}

The online change point detection problem is often formulated as a sequential two-sample test which has been widely considered in the past \cite{page1954continuous, lorden1971procedures, shewhart1986statistical, shiryaev1963optimum}. Given a sequence of samples $\{X_i\}$ from a stationary mixing process $\bd X$ with distribution $\mu$, at each time step, the following null and alternative hypotheses are proposed
\begin{align}
    &H_0: \mu \text{ remains the same},\nn\quad\quad H_1: \mu \text{ has changed}.\nn
\end{align}
Test statistics are calculated using the samples collected up to the current time step. To detect the change quickly and accurately, one attempts to reject the null hypothesis $H_0$ via a threshold rule at every time step.

More formally, consider a stationary stochastic process $\bd X = \{X_i\}_{i\in\bb Z}\in\msf X^\infty$ adapted to its natural filtration with unknown distribution $\mu$. At some unknown but deterministic time index $\tau\in\bb Z$, we have $X_i \sim \mu$ for $0\leq i\leq \tau$ and $X_i \sim \nu$ for $i\geq\tau+1$, where $\mu, \nu\in\calpx$ and $\mu\neq \nu$. This can be conceptually thought of as having a separate and independent stochastic process $\bd X'\in\msf X^\infty$ following unknown distribution $\nu$ running alongside $\bd X$. From the outside, one can only observe $\bd X$ up to time $\tau$, and at time $\tau$, the observation is immediately switched to $\bd X'$. 

Suppose the null hypothesis is rejected at time $T(b)$, which is a stopping time adapted to the filtration $\{\cal X_{-\infty}^i\}_{i\in\bb Z}$ and a function of the threshold $b$. If we use $\bb E_\infty$ and $\bb E_0$ to denote the expectation under $H_0$ and $H_1$ respectively, then the average-run-length $\mtbfa$ and the average-detection-delay $\md$ can be written in terms of the stopping time $T$ as follows
\begin{align}
    \mtbfa = \bb E_\infty[T(b)] \quad \text{ and }\quad \md = \bb E_0[T(b)].\nn
\end{align}
Unlike the Bayesian formulation, we assume the change point $\tau$ is unknown and deterministic, and thus we set $\tau=0$ without loss of generality. $\mtbfa$ measures the robustness of the test against false alarms. Whereas $\md$ measures the quickness of the test in response to an abrupt change. The overall goal of online change point detection is to have a $\mtbfa$ that grows with $b$ as fast as possible and a $\md$ that grows with $b$ as slowly as possible. 

\subsection{MMD-CUSUM test}
The MMD-CUSUM test is a sequential adaptation of the kernel two-sample test. Consider a bounded, measurable, characteristic, reproducing kernel $k: \msf X\times\msf X\to \bb R$. Denote the reference dataset as $\cal D_h = \{ X'_{i}\}_{i=1}^h$ of size $h$. The detection algorithm processes the incoming data in blocks of size $r$, which is denoted as $\cal B_r(t) = \{X_i\}_{i=(t-1)r}^{tr-1}$ for an integer $t \geq 1$. Let $\hat \nu_h$ and $\hat \mu_r$ denote the empirical measure constructed using the dataset $\cal D_h$ and $\cal B_r(t)$. Define the $\mmd$ between these two empirical measures as
\begin{align}
\label{eqn: empirical mmd}
    &\mmd[\hat\mu_r,\hat\nu_h] \nn\\
    = &\bigg(\frac{1}{r^2}\sum_{1\leq n,m\leq r} k(X_{(t-1)r+n-1}, X_{(t-1)r+m-1})\nn\\
    \quad & + \frac{1}{h^2}\sum_{1\leq n,m\leq h} k(X'_{n}, X'_m) \nn\\
    \quad & - \frac{2}{rh} \cdot\sum_{\mathclap{\substack{1\leq n \leq r,\\ 1\leq m\leq h}}} k(X_{(t-1)r+n-1}, X'_{m})\bigg)^{\frac{1}{2}},
\end{align}
At time step $i = t\cdot r$, the algorithm computes the following test statistic; otherwise, it collects the new observations and remains idle. Let integer $M \geq r$ be the minimum number of samples required to perform the test. We write the test statistics at time step $i$ as
\begin{align}
\label{eqn: test statistic}
    \hat{s}_{\lfloor i/r\rfloor} &= \max_{1\leq n \leq t} s_{n:\lfloor i/r\rfloor},\\
    s_{n:\lfloor i/r\rfloor} &= \sum_{t=n}^{\lfloor i/r\rfloor}\bigg\{ \mmd[\hat\mu_r,\hat\nu_h] - \Delta\bigg\},\nn
\end{align}
where $\Delta > 0$ is a tunable parameter that keeps the summand slightly blew 0 under the null hypothesis. The corresponding stopping rule with threshold $b$ and $M$ minimum samples is written as
\begin{align}
\label{eqn: stopping rule}
    T(b, M) &= r\cdot \inf\big\{t\geq M/r : \hat{s}_t > b\big\}.
\end{align}

{We make the following remarks regarding the above MMD-CUSUM statistics. }

\paragraph{Convergence of Empirical MMD}

{
To correctly configure the offset parameter $\Delta$, we need to determine the envelope of the deviation of the empirical MMD from the true one. The result collected in the following lemma shows that the estimation error is bounded by a term diminishing in the sample size plus a small margin, almost surely for all three mixing conditions. Note that the empirical MMD can be equivalently written as the MMD between empirical measures. For probability measures $\mu$ and $\nu$, we write $\widehat{\mmd}(\mu, \nu)$ as $\mmd(\hat \mu_r, \hat \nu_h)$ where $\hat \mu_r, \hat\nu_h$ are empirical measures of $\mu$ and $\nu$ with $r$ and $h$ samples, respectively.
}
\begin{lemma}
\label{lem: empirical mmd consistency}
    Let $\bd X$ and $\bd X'$ be two independent processes with stationary distribution $\mu$ and $\nu$ satisfying the mixing conditions introduced before. Given $\delta > 0$, there exist constant $C(r, h)$ such that the following holds almost surely for sufficiently large $h$,
    \begin{align}
        \bigg|\bb E[\mmd(\hat{\mu}_r, \hat{\nu}_h)| \cal D_h] - \mmd(\mu, \nu)\bigg| \leq C(r, h) + \delta,\nn
    \end{align}
    where $C(r, h) = O\big(\sqrt{\frac{1}{r}} + \sqrt{\frac{\log\log h}{h}}\big)$ and $\bb E[\cdot | \cal D_h]$ denotes the expectation taken over the randomness in $\hat{\mu}_r$ conditioned on the reference dataset $\cal D_h$.
\end{lemma}
\begin{proof}
    Applying triangle inequality, we get the following two expressions:
    \begin{align*}
        \mmd(\hat{\mu}_r, \hat{\nu}_h) - \mmd(\mu, \nu) \leq \mmd(\hat{\mu}_r, \mu) + \mmd(\nu, \hat{\nu}_h),\\
        \mmd(\mu, \nu) - \mmd(\hat{\mu}_r, \hat{\nu}_h) \geq -\mmd(\hat{\mu}_r, \mu) -\mmd(\nu, \hat{\nu}_h).
    \end{align*}
    Let us consider the first inequality above, and the other one follows similarly. Suppose we take expectation over the randomness of $\hat{\mu}_r$, and due to independence we have,
    \begin{align*}
         \bb E&[\mmd(\hat{\mu}_r, \hat{\nu}_h) | \mathcal D_h] - \mmd(\mu, \nu)\\
         &\leq \bb E[\mmd(\hat{\mu}_r, \mu)] + \mmd(\nu, \hat{\nu}_h).
    \end{align*}
    On the right hand side, the term $\bb E_{\bd X}[\mmd(\hat{\mu}_r, \mu)] \leq \sqrt{\frac{1+2\Sigma_\mu}{r}}$ by Lemma 7.1 of \cite{cherief2022finite} for all $r>0$ and $\bd X$ which satisfies $\Sigma_\mu < \infty$. 
    It remains to bound $\mmd(\nu, \hat{\nu}_h)$ for a particular $\hat{\nu}_h$. Observe that 
    \begin{align*}
       \mmd(\nu,\hat{\nu}_h) = h^{-1}\left\|\sum_{i}^r H_i\right\|_{\cal H_k},
    \end{align*}
    where $\{H_i = k(X_i, \cdot) - \bb E_\nu k\}$ is a Hilbert space valued stochastic process and $\{H_i\}$ enjoys the same mixing property as $\bd X'$ since $H_i$ is a measurable function of $X'_i$. Thus, we can apply the law of iterated logarithm for Hilbert space valued $\alpha$-mixing processes \cite[Theorem 6]{dehling1982almost} or \cite[Theorem 2]{merlevede2008maximal} to conclude there exists constant $c_0 > 0$ such that almost surely
    \begin{align*}
        \limsup_{r\to\infty}\bigg\|\sum_{i=1}^r H_i\bigg\|_{\cal H_k} \leq c_0 \sqrt{h\log\log h}.
    \end{align*}
    Note that the hypothesis of \cite[Theorem 6]{dehling1982almost} holds in our case under the assumption of bounded kernel $k$ and exponential $\alpha/\beta$-mixing and fast $\phi-$mixing. Thus, there exists a constant $C(r, h) = O\big(\sqrt{\frac{1}{r}} + \sqrt{\frac{\log\log h}{h}}\big)$ such that $\mmd(\hat{\mu}_r, \hat{\nu}_h) - \mmd(\mu, \nu) \leq C(r, h) + \delta$ for sufficiently large $h$. Similar, $\mmd(\mu, \nu) - \mmd(\hat{\mu}_r, \hat{\nu}_h)$ can bounded from below with $-C(r, h)-\delta$, and the proof is complete. 
\end{proof}
{
    Lemma \ref{lem: empirical mmd consistency} indicates that under the null hypothesis (no change), the bias of empirical MMD is bounded by a positive quantity decaying at rate $o(r^{-1/2} + h^{-1/2}\log\log h)$ plus a small margin for sufficiently large reference data. To maintain a low value of the MMD-CUSUM statistics under the null hypothesis, it is necessary to apply a certain negative offset to the empirical MMD so that the cumulative sum in \eqref{eqn: test statistic} does not blow up when change is absent which leads to the second remark regarding the parameter $\Delta$.
}

\paragraph{Offset parameter $\Delta$}
{
    We shall determine the appropriate range for the offset parameter $\Delta$ in \eqref{eqn: test statistic} using Lemma \ref{lem: empirical mmd consistency}. Note that $\Delta$ needs to be sufficiently large under the null hypothesis such that the MMD-CUSUM statistic does not blow up due to the estimation error of the empirical MMD. As suggested by Lemma \ref{lem: empirical mmd consistency}, if $\Delta$ is strictly larger than $C(r, h)$, i.e., $\Delta \geq C(r, h) + \delta$ for some margin $\delta > 0$, then the empirical MMD is bounded by $\Delta$ almost surely for sufficiently large sample size. On the other hand, the upper bound for $\Delta$ appears under the alternative hypothesis (with post-change distribution $\nu$). As we shall see in Theorem \ref{thm: md bound}, $\Delta$ should be strictly less than $\mmd_k(\mu, \nu) - C(r, h) - \delta$ otherwise the $\md$ can be unbounded. 
}
{
    To tune $\Delta$ in practice, one can simulate the pre-change scenario with different values of $\Delta \geq C(r, h) + \delta$ using the reference dataset. For each value of $\Delta$, the $\mtbfa$ can be estimated with multiple runs of the experiment. Then, choose the smallest $\Delta$ that yields the acceptable $\mtbfa$ performance. Keeping the $\Delta$ small allows the MMD-CUSUM to achieve better $\md$.
}

\section{MAIN RESULTS}
\label{sec: main results}

In this section, we establish the detection performance of the MMD-CUSUM test using the metrics introduced in the previous section. The average-run-length $\mtbfa$ characterizes the average interval between false alarms, which is lower-bounded in Theorem \ref{thm: mtbfa bound}. The average detection delay measures the quickness of the detection, and an upper bound is given in Theorem \ref{thm: md bound}. The proofs are omitted due to the page limit, and they can be found in the Supplementary Material's \hyperref[sec: app]{Appendix}. 

Before we state the results, let us briefly summarize the technique we employed. Recall $\bb E_\infty$ denotes the expectation under $H_0$. We can expand the $\bb E_\infty[T(b,M)]$ as follows
\begin{align}
    &\bb E_\infty [T(b, M)] = \sum_{t=1}^\infty \bb P_\infty[T(b, M) \geq t] \nn\\
    =& M + \sum_{l=M+1}^\infty \left(1 - \mathbb{P}_\infty\left\{\bigcup_{t=M+1}^{l} \{T(b, M) = t\}\right\}\right)\nn\\
    \geq & M + \sum_{l=M+1}^\infty \left(1 - \mathbb{P}_\infty\left\{\bigcup_{t=M+1}^{l}\bigcup_{k=1}^{t-M}\{s_{k:t} \geq b\}\right\}\right) \nn\\
    \geq & M + \sum_{l=M+1}^\infty \left(1 - \sum_{t=M+1}^{l}\sum_{k=1}^{t-M}\mathbb{P}_\infty\{s_{k:t} \geq b\}\right),\nn
\end{align}
where union bound is applied to the last inequality. At this point, it suffices to obtain an upper bound on the tail probability $\mathbb{P}_\infty\{s_{k:t} \geq b\}$ using Proposition \ref{thm: hoeffding for phi mixing}, \ref{thm: hoeffding for beta mixing}. The tail probability bounds in Proposition \ref{thm: hoeffding for phi mixing}, \ref{thm: hoeffding for beta mixing} offers simple explicit subGaussian decay rates with linear or sublinear dependency\footnote{Linear or sublinear dependency of sample size means a tail bound of $O(\exp(-g(n)\epsilon^2))$ where $g(n)$ grows linearly or sublinearly. By writing it this way, we assume the tail probability measures the event $\{\sum_{i=1}^n X_n - n\bb E X \geq n\epsilon\}$ where the deviation scales with sample size $n$.} on the sample size $n$ inside the exponential. This kind of decay rate is necessary for our analysis as it dictates the scaling of $\mtbfa$ in threshold $b$. As we shall see in the theorem below, the slower decay rate of Proposition \ref{thm: hoeffding for beta mixing} causes the difference in $\mtbfa$ between exponential $\alpha/\beta$-mixing and fast $\phi$-mixing processes.

We note that the existing concentration inequalities obtained for generic purposes are not well-suited for the task at hand. For example, the classic concentration inequalities for $\alpha$-mixing, such as \cite[Theorem 3.5]{vidyasagar2002learning}, have tail bound with an additive term in addition to the common exponential term seen in the usual Hoeffding's inequality. When combined with our technique, it leads to a prohibitively cumbersome derivation of the $\mtbfa$. The $\alpha$-mixing concentration inequality in \cite[Theorem 2]{zou2009generalization} gives the tail bound on the relative deviation (scaled by variance) instead of the absolute deviation. The $\beta$-mixing results in \cite{krebs2018large} and the $\alpha$-mixing results in \cite{merlevede2012bernstein} provide a subexponential bound of $O(\exp(-\epsilon))$ which is a weaker dependency on $\epsilon$ than we desired. The detailed discussion of the concentration inequalities we derived is postponed until the main results are introduced.

We now state the main result on the upper bound of $\mtbfa$ under the mixing condition described in Definition \ref{def: exp mixing}. 
\begin{theorem}
\label{thm: mtbfa bound}
    The average-run-length for test statistics \eqref{eqn: test statistic} and stopping rule \eqref{eqn: stopping rule} under the null hypothesis has the following lower bounds.
    \begin{enumerate}
        \item Suppose $\bd X$ is $\alpha/\beta$-mixing satisfying Definition \ref{def: exp mixing}, then
        \begin{align}
        \label{eqn: mtbfa beta}
            &\mtbfa[T(b, M)]\nn\\
            &\quad\geq M -1 + \exp\bigg(b^{\frac{\gamma}{\gamma+1}} \delta^{\frac{\gamma + 2}{\gamma+1}}\bigg) (1 + o(1)),
        \end{align}

        \item Suppose $\bd X$ is $\phi$-mixing satisfying Definition \ref{def: fast phi mixing}, then
        \begin{align}
        \label{eqn: mtbfa phi}
            \mtbfa[T(b, M)] \geq M -1 + \exp(b\delta)(1 + o(1)),
        \end{align}
    \end{enumerate}
    {where $\gamma$ is defined in Definition \ref{def: exp mixing}, and $\delta > 0$ is defined in Lemma \ref{lem: empirical mmd consistency} and depends on $\Delta$, $h$.}
\end{theorem}
\begin{proof}
    Appendix \ref{sec: mtbfa proof}
\end{proof}

Theorem $\ref{thm: mtbfa bound}$ establishes the first $\mtbfa$ bound for the MMD-CUSUM test under $\alpha$/$\beta$/$\phi$-mixing processes. The performance of the MMD-CUSUM test under $\alpha$/$\beta$-mixing case has not been considered in the literature before, and Equation $\ref{eqn: mtbfa beta}$ provides the first exponential lower bound on the $\mtbfa$. In previous studies, $\phi$-mixing processes are considered in certain specific cases, such as the uniformly ergodic Markov chains \cite{chen2022change} and hidden Markov models (HMM) \cite{zhang2023data}. Equation \ref{eqn: mtbfa phi} generalizes the $\mtbfa$ bound therein to the broader $\phi$-mixing processes without loss of performance. It also indicates that Markovian or HMM structures are not necessary for the exponential lower bond of the $\mtbfa$.

The $\mtbfa$ bound in Equation \ref{eqn: mtbfa beta} has a dependency on $\gamma$, which controls the mixing speed (Definition \ref{def: exp mixing}). This dependency on $\gamma$ also is the result of applying the concentration bound in Proposition \ref{thm: hoeffding for beta mixing}. Suppose the $\alpha$ or $\beta$-mixing coefficient has a decay rate of $O(\exp(-n))$, i.e., $\gamma=1$, the $\mtbfa$ then achieves a $\Omega(\exp(b^{1/2} \delta^{3/2}))$ lower bound which is slighted degraded in terms of the threshold $b$ compared to Equation \ref{eqn: mtbfa phi}. 

Surprisingly, the $\mtbfa$ under the fast $\phi$-mixing condition (Equation \ref{eqn: mtbfa phi}) achieves the $\Omega(\exp(b))$ lower bound (same as Markovian samples) while only requiring a summable $\phi$-mixing coefficient. In comparison, the $\mtbfa$ lower bounds in \cite{chen2022change} and \cite{zhang2023data} are obtained under the Doeblin's condition \cite[page 402]{meyn2012markov}, which corresponds to exponential $\phi$-mixing conditions.

To measure the quickness of the MMD-CUSUM test, we estimate the expected value of the stopping time $T(b, M)$ under the alternative ($H_0$). Recall that $\bb E_0$ denotes the expectation under $H_1$. We can write the $\bb E_0[T(b,M)]$ as follows
\begin{align}
        \bb E_0 [T(b, M)] &= \sum_{t=1}^\infty \bb P_0[\hat{s}_t \leq b] \leq \sum_{t=1}^\infty \bb P_0[s_{1: t} \leq  b] \nn\\
        &= \sum_{t=1}^{t_0}\bb P_0[s_{1: t} \leq b] + \sum_{t=t_0+1}^\infty \bb P_0[s_{1: t} \leq  b],\nn
\end{align}
where the first inequality is due to $s_{1:t} \leq \hat s_t$. Splitting the summation at $t_0$ and trivially bound the first term with $t_0$. With a certain choice of $t_0$, the second term can be shown to be ultimately negligible or $o(1)$ compared to $t_0$ using the concentration inequality in Proposition \ref{thm: hoeffding for phi mixing} and \ref{thm: hoeffding for beta mixing}. 

\begin{theorem}
\label{thm: md bound}
    Suppose $\bd X$ is a mixing process satisfies Definition \ref{def: exp mixing} or \ref{def: fast phi mixing} and pre and post change stationary distribution $\mu$ and $\nu$ satisfy {$\mmd_k(\mu, \nu) > C(r,h) + \Delta + \delta$} for some $\delta > 0$. The average-detection-delay for test statistics \eqref{eqn: test statistic} and stopping rule \eqref{eqn: stopping rule} under the alternative hypothesis has the following upper bounds.
    \begin{align}
    \label{eqn: md bound}
        &\md[T(b, M)] \nn\\
        &\leq \max\bigg\{M, \frac{b}{D(\mu, \nu) - \Delta - \delta}\bigg\}(1 + o(1)),
    \end{align}
    where $D(\mu, \nu) = \mmd_k(\mu, \nu)- C(r,h)$, and $C(r,h)$ is defined in Lemma \ref{lem: empirical mmd consistency}.
\end{theorem}
\begin{proof}
    Appendix \ref{sec: md proof}
\end{proof}

Theorem \ref{thm: md bound} gives the first $O(b)$ upper bound on $\md$ under all three types of mixing conditions. Similar to the $\mtbfa$ lower bound, it was previously considered only under uniformly ergodic Markov chains and HMM. Our result shows that the Markovian or HMM structure is also not necessary for $O(b)$ upper bound on $\md$. 

Intuitively, the realization of the MMD-CUSUM statistics should track its mean, which is just $n\mmd(\mu, \nu)$ for time $n$. Therefore, the threshold $b$ should affect the average detection delay in a linear fashion. We note that the sufficient separation between $\mu$ and $\nu$ is required due to the estimation error of the empirical MMD as indicated by Lemma \ref{lem: empirical mmd consistency}. This can be satisfied by choosing $r, h$ sufficiently large and $\delta$ sufficiently small to ensure {$\mmd_k(\mu, \nu)-2C(r,h) - 2\delta > 0$}.

We now establish the concentration inequalities for the sum of bounded functions under mixing conditions. Proposition \ref{thm: hoeffding for phi mixing} is a Hoeffding-type inequality for $\phi$-mixing processes with summable mixing coefficients. We provide a proof based on the martingale decomposition. Concentration inequality for exponential $\phi$-mixing processes is obtained in \cite{samson2000concentration} using an information inequality-based argument. The martingale-based method was used in \cite{Leonid_Kontorovich_Ramanan_2009} to study the concentration inequality of dependent random variables on countable spaces. Proposition \ref{thm: hoeffding for beta mixing} compliments the results therein by considering stationary $\phi$-mixing processes on completely separable metric spaces. The sum of bounded functions of $\phi$-mixing processes has a tight concentration bound that resembles that of i.i.d. random variables, which can be recovered by setting $\Phi=0$. 
\begin{prop}
\label{thm: hoeffding for phi mixing}
    Let $\bd X$ be a stationary $\phi$-mixing process with coefficient satisfying Definition \ref{def: fast phi mixing}. Assume that $f:\msf X\to \bb R$ has bounded span and let $S_n = \sum_{i=0}^if(X_i)$. Then for $\epsilon \geq 0$, it holds
    \begin{align}
        &\bb P\bigg[S_n - \sum_{i=0}^{n-1}\bb{E}[f(X_i)]\geq n\epsilon \bigg] \nn\\
        &\quad\leq \exp\bigg(-\frac{2n\epsilon^2}{(2\Phi + 1)^2\sp(f)^2}\bigg),\nn
    \end{align}
    where $\Phi$ is defined in Definition \ref{def: fast phi mixing}.
\end{prop}
\begin{proof}
    Appendix \ref{sec: proof of hoeff phi}
\end{proof}

Compared to the $O(\exp(-n\epsilon^2))$ tail bound in Proposition \ref{thm: hoeffding for phi mixing}, the following concentration inequality for $\beta$-mixing processes has an $O(\exp(-\hat n\epsilon^2))$ tail bound where $\hat n$ grows sublinearly with the sample size $n$. The proof follows \cite[Theorem 2]{zou2009generalization} with the modification of replacing Bernstein's inequality with Hoeffding's Lemma (Lemma \ref{lem: hoeffding lemma}) to yield the desired result for our purpose.

\begin{prop}
\label{thm: hoeffding for beta mixing}
    Let $\bd{X}$ be a stationary $\beta$-mixing sequence with the coefficient satisfying Definition \ref{def: exp mixing}. Assume that $f:\cal Y\to \bb R$ has bounded span, i.e., $\sp(f) < \infty$, and let $S_n = \sum_{i=0}^{n-1} f(X_i)$. Then for all $\epsilon \in (0, \sp(f))$, it holds
    \begin{align}
        &\bb P\bigg[S_n - \sum_{i=0}^{n-1}\bb{E}[f(X_i)]\geq n\epsilon\bigg] \nn\\
        &\leq (1 + \bar\beta/e^{2})\exp\bigg\{- \frac{2\hat n\epsilon^2}{\sp(f)^2} \bigg\},\nn
    \end{align}
    where $\hat n = \lfloor n\lceil (10n/c)^{1/(\gamma+1)} \rceil^{-1}\rfloor$ and $c, \gamma$ are defined in Definition \ref{def: exp mixing}.
\end{prop}
\begin{proof}
    Appendix \ref{sec: proof of hoeffding for alpha mixing}
\end{proof}

To our knowledge, the tail bound in the above form has not been considered previously. As opposed to the classic two-term version in \cite[Theorem 3.5]{vidyasagar2002learning} and the relative error version in \cite[Theorem 2]{zou2009generalization}, which can be difficult to be applied in our analysis, Proposition \ref{thm: hoeffding for beta mixing} streamlines the calculation of $\mtbfa$ and $\md$ in Theorem \ref{thm: mtbfa bound} and \ref{thm: md bound}. 

Compared to regular Hoeffding's inequality for bounded i.i.d. random variables \cite{hoeffding1994probability}, the exponent of the tail bound has a sublinear dependence on sample size due to the presence of $\hat n$. $\hat n$ is close to $n$ when $\gamma$ is large corresponding to a faster decaying $\beta$-mixing coefficient (Definition \ref{def: exp mixing}). This sublinear relation with respect to $n$ is also reported by \cite{merlevede2012bernstein} and \cite{krebs2018large} as well under exponential $\alpha$ and $\beta$-mixing conditions with $\gamma=1$. They provided an $O(\exp(-n\epsilon/(\log n\log\log n)))$ tail bound, which is a faster rate in $n$ compared to Proposition \ref{thm: hoeffding for beta mixing} with $\gamma=1$. It is tempting to think that this tail could improve the lower bound of $\mtbfa$ in Theorem $\ref{thm: mtbfa bound}$. However, the subexponential, instead of subGaussian\footnote{A subGaussian bound on $\epsilon$ refers to a tail bound that looks like $O(\exp(-\epsilon^2))$. A subGaussian bound on $\epsilon$ refers to a tail bound that looks like $O(\exp(-\epsilon))$.}, dependency on $\epsilon$ makes it not applicable to our proof. A similar concentration type inequality for $\alpha$-mixing processes is obtained in Proposition \ref{thm: hoeffding for alpha mixing} following an analogous proof. 

\begin{figure*}[ht]
    \centering
    \begin{subfigure}{0.49\linewidth}
        \includegraphics[width=\linewidth]{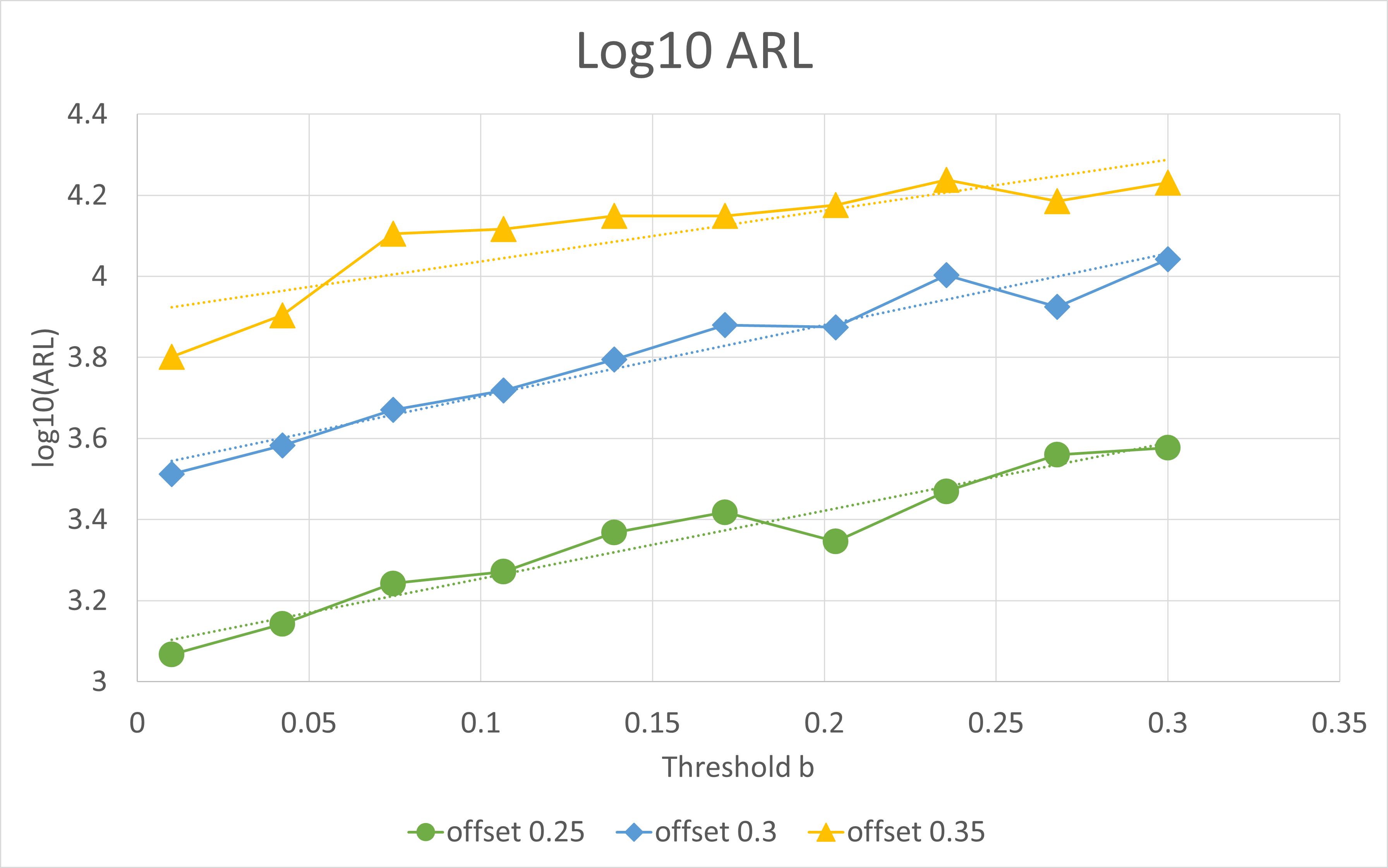}
        \caption{$\log_{10}(\mtbfa)$ with Gaussian noise.}
        \label{fig:log arl normal}
    \end{subfigure}
    \hfill
    \begin{subfigure}{0.49\linewidth}
        \includegraphics[width=\linewidth]{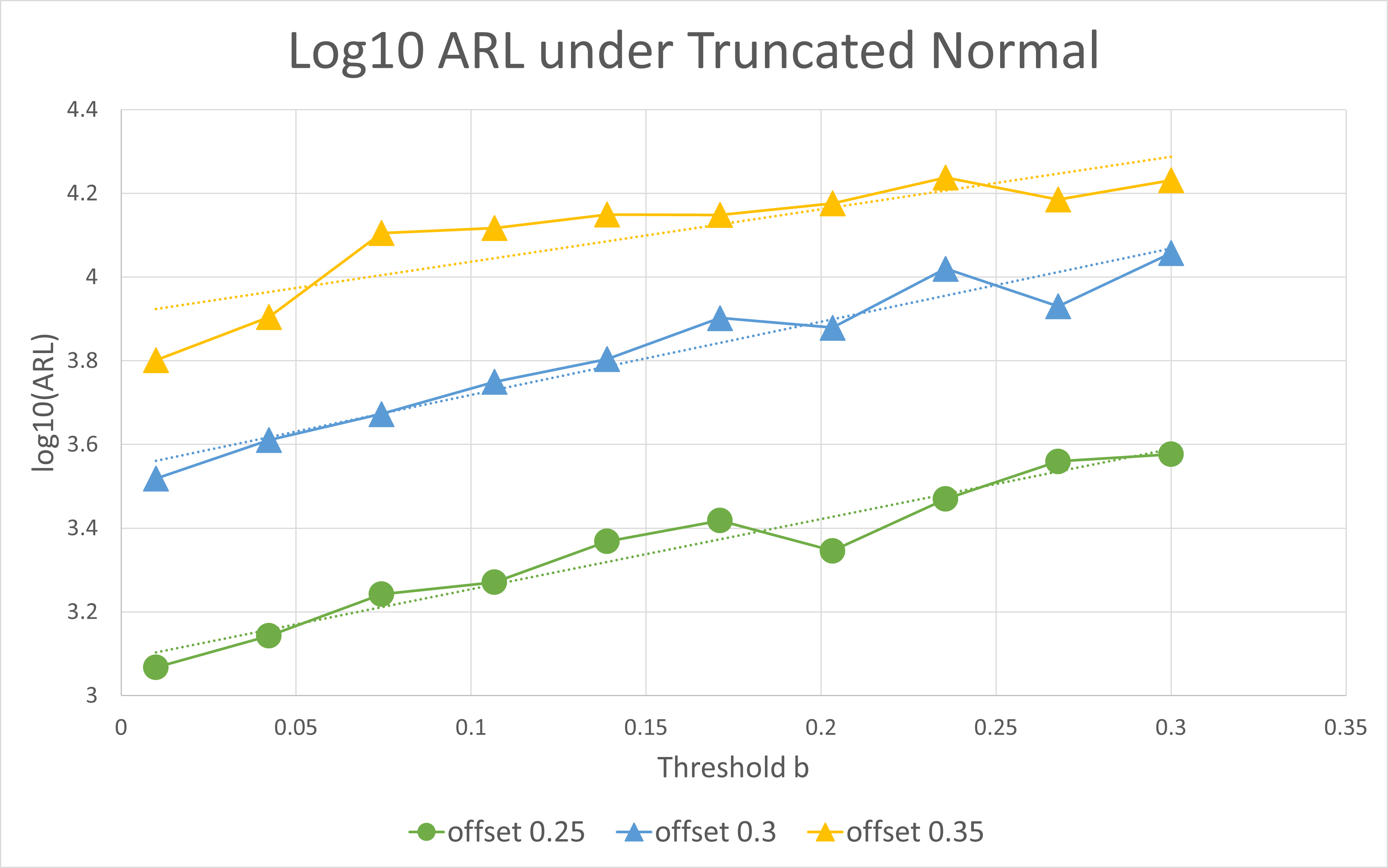}
        \caption{$\log_{10}(\mtbfa)$ with truncated Gaussian noise.}
        \label{fig:log arl truncated normal}
    \end{subfigure}
    \begin{subfigure}{0.49\linewidth}
        \includegraphics[width=\linewidth]{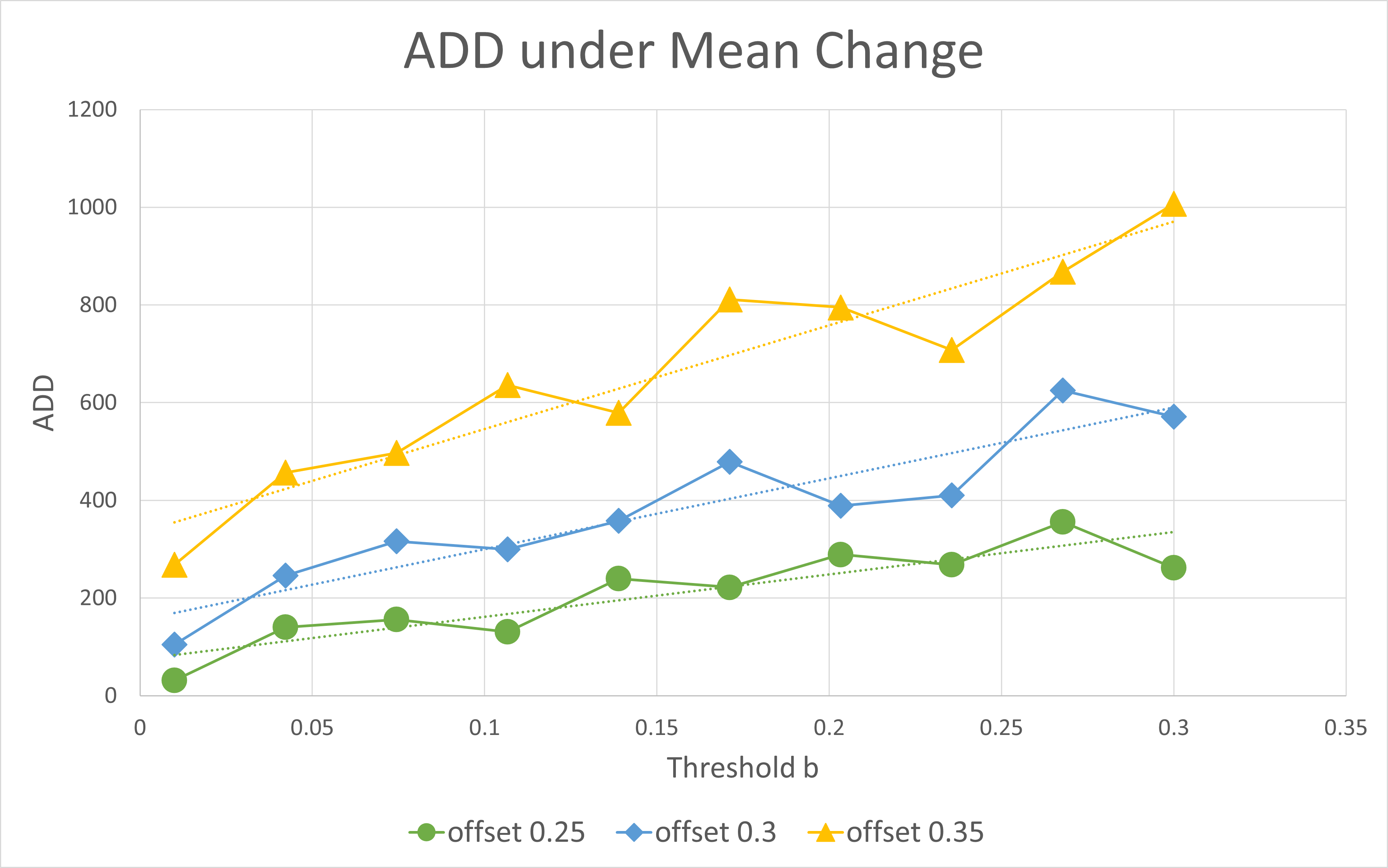}
        \caption{$\md$ under mean shift.}
        \label{fig:add mean}
    \end{subfigure}
    \hfill
    \begin{subfigure}{0.49\linewidth}
        \includegraphics[width=\linewidth]{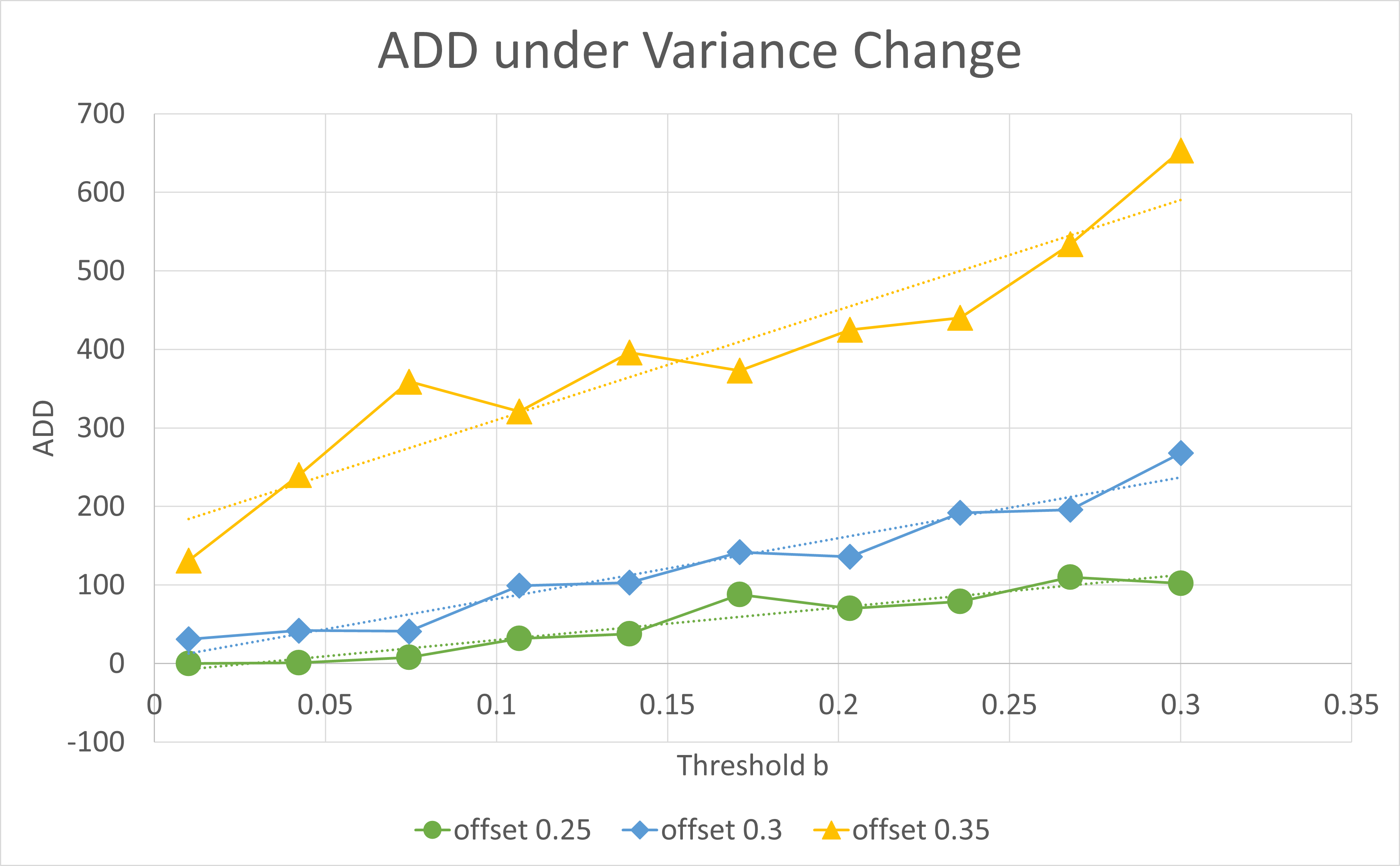}
        \caption{$\md$ under variance change.}
        \label{fig:add var}
    \end{subfigure}
\end{figure*}

\section{NUMERICAL SIMULATIONS}
\label{sec: simulation}

In this section, we apply the MMD-CUSUM test to a simulated stochastic process and verify the theoretical results. The stochastic process is generated by simulating a stable linear system $A\in\bb R^{4\times 4}$ with an observations matrix $C\in\bb R^{2\times 4}$. Let $\bd Z = \{Z_i\}_{i\in\bb N}$ denote the state process and $\bd Y=\{Y_i\}_{i\in\bb N}$ denote the observation process. The system update equations can written as follows
\begin{align}
    Z_{i+1} &= A Z_i + W_i\nn\\
    Y_{i+1} &= C Z_{i+1} + V_i,\nn
\end{align}
where $A=$[[0.96, 0.99, -0.88, 0.56],[0, 0.98, 0.75, -0.65],[0, 0, 0.97, 0.95], [0, 0, 0, 0.94]] and $C = $[[1, 0, 0, 0,], [0, 0, 0, 0], [0, 0, 1, 0], [0, 0, 0, 0]]. Randomness is introduced into the system through the actuation noise $W_i$ and observation noise $V_i$ where $W_i\stackrel{i.i.d.}{\sim} \cal N_1$ and $V_i\stackrel{i.i.d.}{\sim} \cal N_2$ for all $i$. In our experiments, $\cal N_1$ and $\cal N_2$ are two multivariate normal distributions. This is an example of a hidden Markov model (HMM). The state observation joint process $(\bd Z, \bd Y)$ and the state process $\bd Z$ along are Markov chains; however, the observation process $\bd Y$ in general is not. The observation process of this system is exponential $\beta$-mixing. {This can be deduced from the fact that the matrix $A$ is stable and the noise has bounded variance \mbox{\cite[Section 3.5, page 100]{vidyasagar2002learning}}}. To obtain an exponential $\phi$-mixing process from the observations, one can simulate the above system with truncated versions of $\cal N_1$ and $\cal N_2$. 

The kernel chosen for the MMD-CUSUM test is the rational quadratic kernel $k_\sigma^{rq}(x, y) = (1 + (2\sigma)^{-1}\|x-y\|^2)^{-\sigma}$ for $\sigma > 0$ instead the popular Gaussian RBF kernel $k_\sigma^{rbf}(x, y) = \exp(-\|x-y\|^2/(2\sigma^2))$ for $\sigma >0$. As demonstrated by \cite{binkowski2018demystifying}, the rational quadratic kernel is favored over the Gaussian RBF kernel in GAN applications, which indicates its superior performance in separating probability distribution. We fix the parameter $\alpha=1$ for all experiments. The reference dataset is obtained by recording $\bd Y$ for $10^4$ steps under the pre-change configurations with an appropriate burn-in period applied to the samples to maintain stationarity. We estimate the $\mtbfa$ and $\md$ by taking the average of 50 independent experiments for each threshold. The experiments are performed under 3 different offsets to demonstrate the sensitivity of this parameter.

We apply abrupt changes to the noise distribution $\cal N_1$ of the state process. The MMD-CUSUM test is applied to the observation process $\bd Y$ only. The noise observation poses an additional layer of challenge for the detector. To simulate an $\alpha$/$\beta$-mixing process, it suffices to use the regular Gaussian noise. To simulate a $\phi$-mixing, we sample from the same Gaussian distribution and reject the samples falling outside a $[-1, 1]^4$ box. The random seeds are kept the same across the regular and the truncated cases to ensure comparability. The $\log(\mtbfa)$ under both cases are shown in Figure \autoref{fig:log arl normal} and \autoref{fig:log arl truncated normal}. To our surprise, the $\mtbfa$ under the regular Gaussian case maintains an exponential relationship with the threshold, which suggests the $\mtbfa$ bound for $\alpha$/$\beta$-mixing process can be improved. We discuss the difficulty associated with this improvement in Section \ref{sec: discussion}.

The $\md$ are estimated under regular Gaussian noise. We present the $\md$ under two cases: (i) mean shift $\cal N_1(0, 0.1 I) \to \cal N_1'(0.01 \mathbbm{1}, 0.1 I)$ (Figure \autoref{fig:add mean}) and (2) variance change $\cal N_1(0, 0.1 I) \to \cal N_1'(0, 0.5 I)$ (Figure \autoref{fig:add var}). The $\md$ scales linearly with the threshold $b$, which corroborates our findings.

\section{DISCUSSION}
\label{sec: discussion}

\subsection{Unbiased MMD estimator}
The following unbiased estimator of the squared MMD, introduced in \cite{gretton2012kernel}, can also be used to replace Equation \ref{eqn: empirical mmd}. We write the unbiased estimator of the squared MMD between $\mu, \nu$ using $m$ samples from $\mu$ and $n$ samples from $\nu$ as
\begin{align}
    &\overline\mmd_{k}^2(\hat\mu_m, \hat\nu_n) = \frac{1}{m(m-1)} \sum_{1\leq i,j\leq N} k(X_i, X_j) \nn\\
    + &\frac{1}{n(n-1)}\sum_{1\leq i,j\leq M} k(X_i', X_j') - \frac{2}{nm}\sum_{i,j}k(X_i, X_j'),\nn
\end{align}
where $X_i\sim\mu$ and $X_j'\sim\nu$ for $i=1,\cdots, n$ and $j=1,\cdots, m$. We abuse the notation here and write the empirical square MMD as the square MMD between empirical measures, although they are not equivalent to the unbiased estimator. Due to the unbiasedness, it is not always non-negative and thus should be directly plugged into the partial sum with the square root. To adapt $\overline\mmd_k^2$ to the current framework, it suffices to obtain a consistency result such as Lemma \ref{lem: empirical mmd consistency}, and the rest should follow. Consider two independent stochastic processes $\bd X = \{X_i\}$ and $\bd X'=\{X_i'\}$ with stationary distributions $\mu,\nu$ and summable kernel mixing coefficients as in Definition \ref{def: kernel mixing coefficient}. Suppose we use $m$ consecutive samples from $\bd X$ and $n$ consecutive samples from $\bd Y$. Then, we can bound the estimation bias caused by the dependency between samples as follows,
\begin{align}
    &\bigg|\bb E [\overline\mmd_{k}^2(\hat\mu_m, \hat\nu_n)] - \mmd_k^2(\mu, \nu)\bigg|\nn\\
    \leq &\bigg|\bb E [\overline\mmd_{k}^2(\hat\mu_m, \hat\nu_n)] - \bb E [\overline\mmd_{k}^2(\mu, \hat\nu_n)] \bigg| \nn\\ 
    &+ \bigg|\bb E [\overline\mmd_k^2(\mu, \hat\nu_m)] - \mmd_k^2(\mu, \nu)\bigg|\nn\\
    \leq &\frac{\Sigma_\mu}{m} + \frac{\Sigma_\nu}{n},\nn
\end{align}
where the second inequality comes from \cite[Lemma 7.1]{cherief2022finite}, $\Sigma_\mu, \Sigma_\nu$ are defined in Definition \ref{def: kernel mixing coefficient}, and the expectations are taken with respect to the randomness in the samples. After denoting $\overline C_{\mu, \nu}(m, n) \coloneqq \frac{\Sigma_\mu}{m} + \frac{\Sigma_\nu}{n}$ and replacing $C_{\mu, \nu}(m, n)$ with $\overline C_{\mu, \nu}(m, n)$ throughout the paper, the same set of results also holds for the CUSUM statistics defined with the unbiased estimator $\overline\mmd_{k}^2$. 

\subsection{Computation complexity}

The time complexity at each time step is $O(rh)$ where $r$ is the block size on the incoming data, and $h$ is the size of the reference dataset. Compared to the overlapping block design in \cite{chen2022change} with time complexity $O(r^2h)$, the non-overlapping block design here increases the speed at the expense of incurring a constant detection delay. The memory usage here is constant since only the current block and the reference data need to be stored. We present the implementation of the detection procedure in Algorithm \ref{alg: MMD CUSUM}.

\subsection{Connection to HMM}

Hidden Markov models (HMM) cover a wide array of real-world scenarios where the MMD-CUSUM test can be applied. For a comprehensive review of HMM, please refer to \cite{ephraim2002hidden} and the references therein. Change point detection for HMM arises from the monitoring complex dynamic systems \cite{smyth1994hidden}, such as communication networks \cite{salamatian2001hidden}, power plants \cite{kwon1999accident}, healthcare monitoring \cite{al2008automatic}, manufacture process monitoring \cite{li2022improved}, distributed machine learning systems, etc. 

For change detection, HMM can be treated as a mixing process. Consider a Markov chain $\bd X \coloneqq \{X_i\}\subset \msf X$ and its observation process $\bd Y \coloneqq \{Y_i\}\subset\msf Y$, where $\msf Y$ is a complete separable metric space with Borel $\sigma$-algebra $\cal Y$. Define the observation kernel $Q_i:\msf X\times\cal Y\to [0, 1]$ and $Q_i(X_i, A) = \bb P(Y_i\in A | \{Y_t\}_{t=-\infty}^{i-1}, \{X_t\}_{t=-\infty}^{i})$. Then, $\bd Y$ is $\alpha$/$\beta$/$\phi$-mixing as soon as $\bd X$ is $\alpha$/$\beta$/$\phi$-mixing \cite[Theorem 3.12]{vidyasagar2002learning}. 

\subsection{Asymptotic stationary processes}

In practice, many mixing processes may not be strictly stationary but convergent towards a stationary distribution at a certain speed. For example, a Doeblin chain starts from an initial distribution that is different from its stationary distribution. Weak asymptotic stationarity was introduced in \cite{agarwal2012generalization} to study the generalization bound of online algorithms. It combines the convergence to a stationary distribution and $\beta$-mixing into a single condition, which we choose not to include for the sake of simplicity. Instead, we provide a discussion on how to adapt Proposition \ref{thm: hoeffding for phi mixing} to asymptotic stationary processes in the supplementary materials. The adaption of Proposition \ref{thm: hoeffding for beta mixing} follows a similar argument. The intuition is that as long as the process converges sufficiently fast, the concentration of the partial sum will still hold. Thus, the same results on $\mtbfa$ and $\md$ can be extended to asymptotic stationary processes at no cost.

\subsection{Obtain \texorpdfstring{$\Omega(\exp(b))$}{Omega(exp(b))} bound on \texorpdfstring{$\mtbfa$}{ARL} under \texorpdfstring{$\alpha$}{alpha}/\texorpdfstring{$\beta$}{beta}-mixing}

As shown in Figure \ref{fig:log arl normal} and $\ref{fig:log arl truncated normal}$, the difference in $\mtbfa$ between $\alpha$/$\beta$-mixing and $\phi$-mixing is minimal which might indicate a tighter $O(\exp(b))$ bound on $\mtbfa$ under $\alpha$/$\beta$-mixing. This would be an improvement over the $\Omega(\exp(b^{1/\gamma}))$ in Theorem \ref{thm: mtbfa bound} where $\gamma$ controls the mixing speed. However, the difficulties lie in the unavailability (to the best of our knowledge) of a subGaussian tail bound with linear dependency on the sample size $n$ for stationary $\alpha$/$\beta$-mixing processes. This bottleneck is also reported by a recent study \cite{arvanitis2023concentration} on the concentration of kernel density estimator with dependent data. Their findings are limited to $\phi$-mixing processes due to the same issue. Circumventing this bottleneck might require significantly new techniques, which are left as future work.

\section{CONCLUSION}
\label{sec: conclusion}

In this paper, we derive the $\mtbfa$ and $\md$ for the MMD-CUSUM test under three stationary mixing conditions. Under the $\phi$-mixing condition, the performance of the MMD-CUSUM test is shown to match the i.i.d. case and the Markov chain case with uniform ergodicity. As a byproduct, we provide concentration inequalities of the partial sum of bounded functionals under $\alpha$, $\beta$, and $\phi$-mixing processes. To our knowledge, the concentration inequality in Proposition \ref{thm: hoeffding for beta mixing} and the proof of Proposition \ref{thm: hoeffding for phi mixing} are novel. 

We note the limitations of this study and future directions as follows. MMD is known to have a poor separation between probability measures, with differences only in the high-frequency region \cite{sriperumbudur2010hilbert}. The MMD-CUSUM test may experience performance degradation in such scenarios. A recent study \cite{hagrass2022spectral} tackles this problem in the kernel two-sample test setting via kernel spectral regularization. The spectral regularized kernel achieves the optimal minimax separation boundary, which results in an improved sample efficiency compared to the usual kernel two-sample test. Additionally, there have been several other exciting developments on kernel two-sample test \cite{Biggs_Schrab_Gretton_2023, Schrab_Kim_Guedj_Gretton_2023, Schrab_Kim_Albert_Laurent_Guedj_Gretton_2023}. It would be an interesting future direction to adapt those methods to the sequential test setting and analyze their performance. 

Another limitation is that our technique does not exploit the finer structures produced by the max operator over the partial sum. The theory of extremes of random fields \cite{yakir2013extremes} provides handy tools to estimate the probability of events such as $\{\sup_{\theta\in\Theta} S_{\theta} > \epsilon\}$, where $S_\theta$ is the sum of $n$ random variables in the random field and $\Theta$ is an index set, such as integers or real numbers. \cite{li2019scan} has demonstrated the utility of this technique in the i.i.d. case and shown a sharp $\mtbfa$ bound of $O(\exp(b^2))$. However, the extension of this technique has yet to be explored in the non-i.i.d. cases. Additionally, leveraging the martingale property of the MMD-CUSUM statistics with an unbiased estimator and the non-asymptotic version of the law of logarithm for martingales \cite{balsubramani2014sharp} yields another possible route to establish the performance bounds. We plan to investigate these directions in the future.

\section*{APPENDIX}
\label{sec: app}
\subsection{Auxiliary Facts}

\begin{definition}[Total variation metric]
\label{def: tv metric}
    Let $\bb B \coloneqq \{f: \supnorm[f] \leq 1, f:\msf X\to \bb R, f \text{ is } \cal X\text{-measurable}\}$, the total variation metric between probability measures $\mu, \nu\in\cal P(\cal X)$ is written as 
    \begin{align*}
        \tv(\mu, \nu) &\coloneqq \frac{1}{2}\sup_{f\in\bb B}\bigg | \int f d\mu - \int fd\nu \bigg |\nn\\
        &= \sup_{A\in\cal X}|\mu(A) - \nu(A)|.
    \end{align*}
\end{definition}

\begin{lemma}[Corollary D.2.5 in \cite{douc2018markov}]
\label{lem: tv ineq}
    Let $f: \msf X\to \bb R$ be an essentially bounded measurable function. For $\mu, \nu\in\cal P(\cal X)$, we have
    \begin{align}
        |\mu(f) - \nu(f)| \leq \tv(\mu, \nu) \sp(f),
    \end{align}
    where $\mu(f)$, $\nu(f)$ denotes the expectation of $f$ under $\mu$, $\nu$.
\end{lemma}

\begin{lemma}
\label{lem: phi mixing ineq}
    Suppose $\{X_i\}$ is a stationary $\phi$-mixing process. Let $g:\msf X^\infty\to \bb R$ be an essentially bounded function and is measurable with respect to the $\sigma$-algebra $\cal X_{t+n}^\infty$. Then
    \begin{align}
        \bigg | \bb E[g(X_{t+n}^\infty) | x_{-\infty}^t] - \bb E[g(X_{t+n}^\infty) | y_{-\infty}^t]\bigg | \leq 2 \phi(n) \sp(g),\nn
    \end{align}
    where $x_{-\infty}^t, y_{-\infty}^t$ are two realizations of the trajectory up to time $t$, and $\sp(g) \leq \supnorm[g]$ when $g$ is non-negative.
\end{lemma}
\begin{proof}
    \begin{align}
        &\bigg | \bb E[g(X_{t+n}^\infty) | x_{-\infty}^t] - \bb E[g(X_{t+n}^\infty) | y_{-\infty}^t]\bigg| \nn\\
        \leq& \bigg | \bb E[g(X_{t+n}^\infty) | x_{-\infty}^t] - E[g(X_{t+n}^\infty)] \bigg|\nn\\
        +& \bigg |E[g(X_{t+n}^\infty) - \bb E[g(X_{t+n}^\infty) | y_{-\infty}^t]\bigg|\nn\\
        \leq& 2\phi(n)\sp(g),\nn
    \end{align}
    where the first inequality is due to triangular inequality and the second is due to the Definition \ref{def: phi mixing} of $\phi$-mixing and Lemma \ref{lem: tv ineq}. 
\end{proof}

\begin{lemma}[Corollary 2.2 in \cite{vidyasagar2002learning}]
\label{lem: alpha mixing ineq}
    Suppose $\{X_i\}$ is a stationary $\alpha$-mixing process. Suppose $g_0, ..., g_l$ are essentially bounded functions, where $g_i$ depends only on $X_{ik}$. Then
    \begin{align}
        \bigg | \bb E \bigg [ \prod_{i=1}^l g_i\bigg] - \prod_{i=1}^l \bb E(g_i)\bigg | \leq 4l\alpha(k)\prod_{i=1}^l\sp(g_i),\nn
    \end{align}
    where $\sp(g_i) \leq \supnorm[g_i]$ when $g_i$ is non-negative.
\end{lemma}

\begin{lemma}[Theorem 2.1 in \cite{vidyasagar2002learning}]
\label{lem: beta mixing ineq}
    Suppose $\{X_i\}$ is a stationary $\beta$-mixing process. Suppose $g_0, ..., g_l$ are essentially bounded functions, where $g_i$ depends only on $X_{ik}$. Then
    \begin{align}
        \bigg |\bb E \bigg [ \prod_{i=1}^l g_i\bigg] - \prod_{i=1}^l \bb E(g_i)\bigg| \leq l\beta(k)\prod_{i=1}^l\sp(g_i),\nn
    \end{align}
    where $\sp(g_i) \leq \supnorm[g_i]$ when $g_i$ is non-negative.
\end{lemma}

\begin{lemma}[Lemma 8.1 in \cite{devroye2013probabilistic}]
\label{lem: hoeffding lemma}
    Let $X$ be a random variable such that $a\leq X\leq b$ almost surely. Then, for $r > 0$,
    \begin{align}
        \bb E[\exp(r(X - \bb E X))] \leq \exp [r^2(b-a)^2/8].\nn
    \end{align}
\end{lemma}

\subsection{Proof of Proposition \ref{thm: hoeffding for phi mixing}}
\label{sec: proof of hoeff phi}

We show a generalized version of Proposition \ref{thm: hoeffding for phi mixing} with time-dependent functions in Proposition \ref{thm: hoeffding for phi mixing gen}, {that is, the partial sum $S_n$ of interest is replaced by $\tilde S_n = \sum_{i=0}^{n-1} f_i(X_i)$ where $f_i$ are potentially different.} The key technique employed here is the martingale decomposition of the partial sum process generated by any stochastic process. In Lemma \ref{lem: martingale decomp of mixing}, we demonstrate the martingale decomposition. In Lemma \ref{lem: span of g_i}, we establish that the martingale difference is bounded under the $\phi$-mixing condition in Definition \ref{def: fast phi mixing}. Finally, we give the proof of Proposition \ref{thm: hoeffding for phi mixing gen} using the two supporting lemmas. 

\begin{prop}
\label{thm: hoeffding for phi mixing gen}
    Let $\bd X$ be a stationary $\phi$-mixing process with coefficient satisfying Definition \ref{def: fast phi mixing}. Assume that $f_i:\msf X\to \bb R$ has bounded span for $i=0, \cdots, n-1$ and let $\tilde S_n = \sum_{i=0}^{n-1}f_i(X_i)$. Then for $\epsilon \geq 0$, it holds
    \begin{align}
        \bb P\bigg[\tilde S_n - \sum_{i=0}^{n-1}\bb{E}[f_i(X_i)] > n\epsilon\bigg] \leq \exp\bigg(-\frac{2n^2\epsilon^2}{\sum_{i=-1}^{n-1} A_i^2}\bigg),
    \end{align}
    where $\Phi$ is defined in Definition \ref{def: fast phi mixing} and $\{A_0, \cdots, A_{n-1}\}$ is defined in Equation \ref{eqn: dummy variable}.
\end{prop}

First, we give the martingale decomposition of the partial sum process generated by any stochastic process.

\begin{lemma}
\label{lem: martingale decomp of mixing}
    Let $\bd X = \{X_i\}$ be a stationary stochastic process on the probability space $(\msf X^\infty, \cal X^\infty, \tbb P)$. For $i\in\{0, \cdots, n-1\}$ and $n\in\bb Z_+$, let $f_i: \bb X\to \bb R$ be a essentially bounded function. Let $\tilde S_n \coloneqq \sum_{i=1}^{n-1} f_i(X_i)$ be the partial sum. Let $\{\cal X_{-\infty}^i\}_{i=0}^{n-1}$ be the filtration generated by $\mbf X$, i.e., $\cal X_{-\infty}^i = \sigma(X_{\infty}, \cdots, X_0, \cdots, X_i)$. Then, there exists a martingale difference sequences $\{D_i\}_{i=1}^{n-1}$ adapted to $\{\cal X_{-\infty}^i\}_{i=0}^{n-1}$ such that 
    \begin{align}
    \label{eqn: martingale decomp of mixing}
        &\tilde S_n - \sum_{i=0}^{n-1}\bb{E}[f_i(X_i)]\nn\\
        =& \sum_{i=0}^{n-1} D_i + \sum_{i=0}^{n-1} \bb E[f_i(X_i)| \cal X_{-\infty}^{-1}] - \sum_{i=0}^{n-1}\bb{E}[f_i(X_i)],
    \end{align}
    where the expectations are taken w.r.t. the stationary distribution.
\end{lemma}

\begin{proof}
    We employ the martingale decomposition technique introduced in Chapter 23 of \cite{douc2018markov}. For the following development, we use these short notations for tuples: $X_{m}^{n} \coloneqq (X_m, \cdots, X_{n})$ and $x_{m}^{n} \coloneqq (x_m, \cdots, x_{n})$, for $0\leq m \leq n$. Without loss of generality, we assume $\bb{E}[f_i(X_i)] = 0$ for $i=0, \cdots, n-1$ to simplify the notation. For $i\in\{-1,0, \cdots, n-1\}$, we define
    \begin{align}
    \label{eqn: definition of g_i}
        g_i (x_{-\infty}^i) \coloneqq \sum_{l=0}^{i} f_l(x_l) + \sum_{l=i+1}^{n-1} \bb{E}[f_l(X_l)| x_{-\infty}^i],
    \end{align}
    where $g_{n-1}(x_{-\infty}^{n-1}) = \sum_{l=0}^{n-1}f_l(x_l)$ and $g_{-1}(x_{-\infty}^{-1}) = \sum_{l=0}^{n-1} \bb{E}[f_l(X_l)| x_{-\infty}^{-1}]$. With this definition, we observe that
    \begin{align}
        g_{n-1}(x_{-\infty}^{n-1}) = \sum_{i=1}^{n-1}\bigg[g_i(x_{-\infty}^i) - g_{i-1}(x_{-\infty}^{i-1})\bigg] + g_0(x_0),\nn
    \end{align}
    and for $i\in\{0, \cdots, n-1\}$ and $x_{-\infty}^{i-1}\in\msf{X}^{i}_{-\infty}$,
    \begin{align}
        g_{i-1}(x_{-\infty}^{i-1}) = \int_{\bb{X}^i}g_i(x_{-\infty}^{i-1}, x_i) d\tbb P_{i}(\cdot| x_{-\infty}^i).\nn
    \end{align}
    Recall that $\cal X_{-\infty}^i = \sigma(X_{-\infty}, \cdots, X_{0}, \cdots, X_{i})$, the above equation shows that $g_{i-1}(X_{-\infty}^{i-1}) = \bb{E}[g_i(X_{-\infty}^i) | \cal X_{-\infty}^{i-1}]$ $\tbb{P}$-a.s. for $i\geq 0$. Thus, $\{g_i(X_{-\infty}^i)\}_{i=0}^{n-1}$ is a $\tbb P$-martingale adapted to filtration $\{\cal X_{-\infty}^i\}_{i=-1}^{n-1}$. It follows the martingale decomposition of $\Tilde{S}_n$ centered at $\sum_{i=0}^{n-1}\bb{E}[f_i(X_i)]$,
    \begin{align}
        &\Tilde{S}_n - \sum_{i=0}^{n-1}\bb{E}[f_i(X_i)] = g_{n-1}(X_{-\infty}^{n-1}) - \sum_{i=0}^{n-1}\bb{E}[f_i(X_i)] \nn\\
        =& \sum_{i=0}^{n-1} D_i + g_{-1}(X_{-\infty}^{-1}) - \sum_{i=0}^{n-1}\bb{E}[f_i(X_i)],\nn
    \end{align}
    where $\{D_i\}_{i=0}^{n-1} = \{g_i(X_{-\infty}^i) - g_{i-1}(X_{-\infty}^{i-1})\}_{i=1}^{n-1}$ is a martingale difference sequence. We arrive at the Equation \ref{eqn: martingale decomp of mixing}.
\end{proof}

Next, we show $D_i$ is bounded by showing $x \mapsto g_i(x_{-\infty}^{i-1}, x)$ has bounded span for any $x_{-\infty}^{i-1}\in \msf X_{-\infty}^{i-1}$ and $i=0, \cdots, n-1$.

\begin{lemma}
\label{lem: span of g_i}
    Suppose $\bd X$ is a stationary $\phi$-mixing process. For each $i\in\{-1, \cdots, n-1\}$, if we define 
    \begin{align}
    \label{eqn: dummy variable}
        A_i \coloneqq 2 \Phi \max\{\sp(f_{l}): i+1 \leq l \leq n-1\} + \sp(f_i),
    \end{align}
    where $A_{-1} \coloneqq 2 \Phi \max\{\sp(f_{l}): 0 \leq l \leq n-1\}$. Then it holds that for $x_{-\infty}^i \in \msf X^{i}_{-\infty}$ and $i=-1, \cdots, n-1$
    \begin{align}
        \label{eqn: g span bound}
        \inf_{x\in\msf X} g_i(x_{-\infty}^{i-1}, x) &\leq g_i(x_{-\infty}^i) \leq \inf_{x\in\msf X} g_i(x_0^{i-1}, x) + A_i,
    \end{align}
    
\end{lemma}
\begin{proof}
     It suffices to show Equation \ref{eqn: g span bound} holds. The first inequality is obvious. To show the second inequality, we pick arbitrary $x^*\in\bb X$ and we have
    \begin{align*}
        &\quad g_i(x_{-\infty}^i) = \sum_{l=0}^{i} f_l(x_l) + \sum_{l=i+1}^{n-1} \bb{E}[f_l(X_l)| x_{-\infty}^i]\\
        &\leq \sum_{l=0}^{i-1}f_l(x_l) +f_i(x^*) + \sum_{l=i+1}^{n-1} \bb{E}[f_l(X_l)| x_{-\infty}^i] \\
        &- \sum_{l=i+1}^{n-1} \bb{E}[f_l(X_l)| x_{-\infty}^{i-1}, x^*] + \sum_{l=i+1}^{n-1} \bb{E}[f_l(X_l)| x_{-\infty}^{i-1}, x^*] \\
        &+ \sp(f_i).
    \end{align*}
    Due to Lemma \ref{lem: tv ineq} (first inequality) and triangular inequality (second inequality), we can see that
    \begin{align*}
         &\quad \sum_{l=i+1}^{n-1} \bb{E}[f_l(X_l)| x_{-\infty}^i] -  \sum_{l=i+1}^{n-1} \bb{E}[f_l(X_l)| x_{-\infty}^{i-1}, x^*] \\
         &\leq \sum_{l=i+1}^{n-1} \sp(f_l) \tv(\tbb P_{l}(\cdot| x_{-\infty}^i), \tbb P_{l}(\cdot| x_{-\infty}^{i-1}, x^*))\\
         &\leq \sum_{l=i+1}^{n-1} \sp(f_l) \bigg[\tv(\tbb P_{l}(\cdot| x_{-\infty}^i), \mu) + \tv(\tbb P_{l}(\cdot| x_{-\infty}^{i-1}, x^*), \mu) \bigg]\\
         &\leq 2 \Phi \max\{\sp(f_{l}): i+1 \leq l \leq n-1\},
    \end{align*}
    where the last inequality follows from Lemma \ref{lem: phi mixing ineq} and Definition \ref{def: fast phi mixing}. Thus, we have
    \begin{align*}
        &g_i(x_0^i)\leq \sum_{l=0}^{i-1}f_l(x_l) +f_i(x^*) + \sum_{l=i+1}^{n-1} \bb{E}[f_l(X_l)| x_{-\infty}^{i-1}, x^*] \\
        +&2 \Phi \max\{\sp(f_{l}): i+1 \leq l \leq n-1\} + \sp(f_i)\\
        \leq &g_i(x_0^{i-1}, x^*) + 2\Phi \max\{\sp(f_{l}): i+1 \leq l \leq n-1\}\\
        + &\sp(f_i)\\
        = &g_i(x_0^{i-1}, x^*) + A_i.
    \end{align*}
    Since $x^*$ is arbitrary, we obtain $g_i(x_0^i) \leq \inf_{x^*\in\bb X} g_i(x_0^{i-1}, x^*) + A_i$, which completes the proof of Equation \ref{eqn: g span bound}.
\end{proof}

Finally, we estimate the tail bound using the classic Chernoff's bounding method  \cite{chernoff1952measure}. In Remark \ref{rem: asymp stationary proof}, we note that a similar tail bound can be obtained for asymptotically stationary processes with a certain convergence (to stationary distribution) rate. 

\begin{proof}
     (\textbf{Proposition \ref{thm: hoeffding for phi mixing gen}}) By Lemma \ref{lem: martingale decomp of mixing}, we have the martingale decomposition of $\bd X$ as in Equation \ref{eqn: martingale decomp of mixing},
    \begin{align*}
        &\quad\tilde S_n - \sum_{i=0}^{n-1}\bb{E}[f_i(X_i)]\\
        &= \sum_{i=0}^{n-1} D_i + \sum_{i=0}^{n-1} \bb E[f_i(X_i)| \cal X_{-\infty}^{-1}] - \sum_{i=0}^{n-1}\bb{E}[f_i(X_i)],
    \end{align*}
    where we abuse the notation and denote $D_{-1} = \sum_{i=0}^{n-1} \bb E[f_i(X_i)| \cal X_{-\infty}^{-1}] - \sum_{i=0}^{n-1}\bb{E}[f_i(X_i)]$.
    One can use Chernoff's bounding method to obtain an exponential bound on the desired quantity. Taking the moment generating function on both sides of Equation \ref{eqn: martingale decomp of mixing} and applying the chain rule for conditional expectation recursively yield for $\theta\geq 0$
    \begin{align}
    \label{eqn: moment gen of tilde s decomp}
        &\bb{E}\bigg[\exp\bigg(\theta\bigg(\Tilde{S}_n - \sum_{i=0}^{n-1}\bb{E}[f_i(X_i)]\bigg)\bigg)\bigg] \nn\\
        =& \bb{E}\bigg[\exp\bigg(\theta\sum_{i=-1}^{n-1} D_i\bigg)\bigg]\nn\\
        =& \bb{E}[\exp(\theta D_{-1})] \prod_{i=0}^{n-1}\bb{E}[\exp(\theta D_i)| \ALP F_{i-1}].
    \end{align}
   
    From lemma \ref{lem: span of g_i}, we know that $D_i$ lies in an interval of length $A_i$ for all $i\in\{-1, \cdots, n-1\}$. By Hoeffding's Lemma \cite[Lemma 23.1.4]{douc2018markov} for bounded martingale difference sequences, we have for $\theta \geq 0$,
    \begin{align}
        \bb{E}[\exp(\theta D_{-1})] &\leq \exp(\theta^2A_{-1}^2/8)\nn\\
        \bb{E}[\exp(\theta D_i)| \ALP F_{i-1}] &\leq \exp(\theta^2A_i^2/8),\nn
    \end{align}
    and plugging above into Equation \ref{eqn: moment gen of tilde s decomp} yields
    \begin{align}
        \bb{E}\bigg[\exp\bigg(\theta\bigg(\Tilde{S}_n - \sum_{i=0}^{n-1}\bb{E}[f_i(X_i)]\bigg)\bigg)\bigg] &\leq \exp\bigg(\frac{\theta^2}{8}\sum_{i=-1}^{n-1} A_i^2\bigg).\nn
    \end{align}
    Applying Markov's inequality to the left-hand side, we have
    \begin{align}
        &\tbb{P}\bigg[\bigg(\Tilde{S}_n - \sum_{i=0}^{n-1}\bb{E}[f_i(X_i)]\bigg) > n\epsilon\bigg] \nn\\
        &\leq \exp(-n\epsilon\theta)\bb{E}\bigg[\exp\bigg(\theta\bigg(\Tilde{S}_n - \sum_{i=-1}^{n-1}\bb{E}[f_i(X_i)]\bigg)\bigg)\bigg] \nn\\
        &\leq \exp\bigg(-n\epsilon\theta + \frac{\theta^2}{8}\sum_{i=-1}^{n-1} A_i^2\bigg).\nn
    \end{align}
    Picking $\theta = 4n\epsilon / \sum_{i=-1}^{n-1}A_i^2$ minimizes the right-hand side and yields
    \begin{align}
        \tbb{P} \bigg[\bigg(\Tilde{S}_n - \sum_{i=0}^{n-1}\bb{E}[f_i(X_i)]\bigg) > n\epsilon\bigg] \leq \exp\bigg(-\frac{2n^2\epsilon^2}{\sum_{i=-1}^{n-1} A_i^2}\bigg).\nn
    \end{align}
    The tail probability of the other side can be bounded analogously. Therefore, we have
    \begin{align}
        \label{eqn: tail bound of alternative center}
        \tbb{P} \bigg[\bigg|\Tilde{S}_n - \sum_{i=0}^{n-1}\bb{E}[f_i(X_i)]\bigg| > n\epsilon\bigg] \leq 2\exp\bigg(-\frac{2n^2\epsilon^2}{\sum_{i=-1}^{n-1} A_i^2}\bigg).
    \end{align}
    This completes the proof after noticing that $\bd X$ is stationary and hence $\bb{E}[f_i(X_i)] = \mu(f_i)$ for $i=0, \cdots, n-1$.
\end{proof}

\begin{remark}
\label{rem: asymp stationary proof}
    For asymptotically stationary processes, the marginal distribution of $X_i$ differs from the stationary distribution $\mu$ but converges to $\mu$ as $i\to\infty$. To consider the tail probability of $\Tilde{S}_n$ centered around $\sum_{i=0}^{n-1}\mu(f_i)$, we apply triangular inequality to $(a)$ and Equation \ref{eqn: tail bound of alternative center} to $(b)$ and obtain
    \begin{align}
        &\tbb{P} \bigg[\bigg|\tilde{S}_n - \sum_{i=0}^{n-1}\mu(f_i) \bigg| > n\epsilon\bigg] \nn\\
        &\stackrel{(a)}{\leq} \tbb{P} \bigg[\bigg|\tilde{S}_n - \sum_{i=0}^{n-1}\bb{E}[f_i(X_i)]\bigg| + \sum_{i=0}^{n-1} \sp(f_i)2\tv(\tbb P_i, \mu) > n\epsilon\bigg]\nn\\
        &\stackrel{(b)}{\leq} 2\exp\bigg(-\frac{2[n\epsilon - \sum_{i=0}^{n-1} \sp(f_i)2\tv(\tbb P_i, \mu)]^2}{\sum_{i=0}^{n-1} A_i^2}\bigg),\nn
    \end{align}
    for $\epsilon \geq n^{-1}\sum_{i=0}^{n-1} \sp(f_i)2\tv(\tbb P_i, \mu)$. Thus a similar tail bound can be obtained after assuming there exists a constant upperbound for $\sum_{i=0}^{n-1} \sp(f_i)2\tv(\tbb P_i, \mu)$ for all $n$.
\end{remark}

\subsection{Proof of Proposition \ref{thm: hoeffding for beta mixing}}
\label{sec: proof of hoeffding for alpha mixing}

We modify the proof of \cite[Theorem 2]{zou2009generalization} by replacing Bernstein's inequality with Hoeffding's Lemma (Lemma \ref{lem: hoeffding lemma}) in bounding the $\Phi_1$ term to yield the desired result for our purpose.

\begin{proof}
    Given integer $n$, choose any integer $k\leq n$ and define $l = \lfloor n/k\rfloor$. Let $p = n-kl$ and define the index sets $I_i$ for $i=1,2,..., k$ as follows
    \begin{align*}
        I_i = \begin{cases}
            \{i, i+k,..., i+lk\} & 1\leq i\leq p\\
            \{i, i+k, ..., i+(l-1)k\} & p+1 \leq i \leq k.
        \end{cases}
    \end{align*}
    Note that $\cup_i I_i = \{1, ..., n\}$ and within each set $I_i$ the elements are pairwise separated by at least $k$.
    Let $G_i = f(X_i) - \bb E[f(X_1)]$, $T(i) = \sum_{j\in I_i} G_j$, and $p_i = |I_i|/ n$ then
    \begin{align*}
        S_n - \bb E S_n = \sum_{i=1}^n G_i = \sum_{i=1}^k \sum_{j\in I_i} T(i)\\
        \sum_{i=1}^k p_i = \frac{1}{n}\sum_{i=1}^k |I_i| = 1.
    \end{align*}
    Now, we write the moment generating function of $\sum_{i=1}^n G_i/n$ for $r \geq 0$, which can be bounded as follows using the convexity of the exponential function,
    \begin{align}
    \label{eqn: mgf bound}
        \bb E\bigg[\exp\bigg (r\frac{\sum_{i=1}^n G_i}{n}\bigg)\bigg] \leq \sum_{i=1}^k p_i \bb E\bigg[\exp\bigg(r\frac{T(i)}{|I_i|}\bigg)\bigg].
    \end{align}
    We now bound the right-hand side in the following fashion. For $i=1, 2,...,k$, we have
    \begin{align}
    \label{eqn: mgf decomposition}
        & \bb E\bigg[\exp\bigg(r\frac{T(i)}{|I_i|}\bigg)\bigg] = \bb E\bigg[\prod_{j=1}^{|I_i|}\exp\bigg(r\frac{G_j}{|I_i|}\bigg)\bigg] \nn\\
        \leq & \underbrace{\prod_{j=1}^{|I_i|} \bb E \bigg[\exp\bigg(r\frac{G_j}{|I_i|}\bigg)\bigg]}_{\Phi_1} \nn\\
        + &\underbrace{\bigg |\bb E\bigg[\prod_{j=1}^{|I_i|}\exp\bigg(r\frac{G_j}{|I_i|}\bigg)\bigg] -  \prod_{j=1}^{|I_i|} \bb E\bigg[\exp\bigg(r\frac{G_j}{|I_i|}\bigg)\bigg]\bigg|}_{\Phi_2}.
    \end{align}
    For convenience, we denote the first term on the right-hand side of the above as $\Phi_1$ and the second term as $\Phi_2$. We bound them separately. $\Phi_1$ can be estimated with Hoeffding's Lemma (Lemma \ref{lem: hoeffding lemma}) for bounded random variables. For $r > 0$,
    \begin{align}
    \label{eqn: phi 1}
        \Phi_1 &= \prod_{j=1}^{|I_i|} \bb E \bigg[\exp\bigg(r\frac{G_j}{|I_i|}\bigg)\bigg]
        \stackrel{(a)}{=} \bigg\{\bb E\bigg[\exp\bigg(r\frac{G_j}{|I_i|}\bigg)\bigg]\bigg\}^{|I_i|}\nn\\
        &\stackrel{(b)}{\leq} \exp\bigg[\frac{r^2\sp(f)^2}{8|I_i|}\bigg], 
    \end{align}
    where $(a)$ is due to stationarity and $(b)$ comes from Lemma \ref{lem: hoeffding lemma}. Note that $|I_i| \geq l$ for $i=1,2,...,k$, thus we have
    \begin{align*}
        \Phi_1 \leq \exp\bigg[\frac{r^2\sp(f)^2}{8l}\bigg].
    \end{align*}
    $\Phi_2$ can be bounded by the $\beta$-mixing inequality in Lemma \ref{lem: beta mixing ineq} and the exponential $\beta$-mixing condition (Definition \ref{def: exp mixing}). 
    \begin{align}
    \label{eqn: phi 2}
        \Phi_2 &= \bigg |\bb E\bigg[\prod_{j=1}^{|I_i|}\exp\bigg(r\frac{G_j}{|I_i|}\bigg)\bigg] -  \prod_{j=1}^{|I_i|} \bb E\bigg[\exp\bigg(r\frac{G_j}{|I_i|}\bigg)\bigg]\bigg |\nn\\
        &\stackrel{(a)}{\leq} \beta(k)(|I_i|-1)\prod_{j=1}^{|I_i|}\bigg\|\exp\bigg[\frac{rG_j}{|I_j|}\bigg]\bigg\|_\infty\nn\\
        &\stackrel{}{\leq} \beta(k)(|I_i|-1)e^{r\sp(f)}\nn\\
        &\stackrel{(b)}{\leq} e^{|I_i| -2} \beta(k) e^{-ck^\gamma} e^{r\sp(f)}\nn\\
        &\stackrel{}{\leq} \frac{\bar\beta}{e^{2}} \exp\{|I_i| + r\sp(f) - ck^\gamma\},
    \end{align}
    where $(a)$ is due to Lemma \ref{lem: beta mixing ineq} and $(b)$ is due to Definition \ref{def: exp mixing} and the fact $\|I_i\|-1\leq \exp(|I_i|-2)$.

    Now, we plug Equation \eqref{eqn: phi 1} and \eqref{eqn: phi 2} into Equation \eqref{eqn: mgf decomposition}, 
    \begin{align*}
        &\bb E\bigg[\exp\bigg(r\frac{T(i)}{|I_i|}\bigg)\bigg] \nn\\
        &\leq \exp\bigg[\frac{r^2\sp(f)^2}{8l}\bigg] + \frac{\bar\beta}{e^{2}} \exp\{|I_i| + r\sp(f) - ck^\gamma\}.
    \end{align*}
    Since $|I_i|$ and $k$ are free variables, we add some structure to simplify the right-hand side of the above. First, we require $4|I_i| \geq r\sp(f)$ which leads to $\exp\{|I_i| + r\sp(f) - ck^\gamma\} \leq \exp\{5|I_i| - ck^\gamma\}$. Next, we require $\exp\{5|I_i| - ck^\gamma\} \leq 1$, which holds if $5|I_i| \leq ck^\gamma$. Since $|I_i|\leq (n/k + 1)$ and $n + k \leq 2n$, it suffices to let $k = \lceil (10n/c)^{1/(\gamma+1)} \rceil$. Then, we have
    \begin{align*}
        \bb E\bigg[\exp\bigg(r\frac{T(i)}{|I_i|}\bigg)\bigg] \leq \exp\bigg[\frac{r^2\sp(f)^2}{8l}\bigg] + \frac{\bar\beta}{e^{2}},
    \end{align*}
    which holds for $0 < r \leq \frac{4l}{\sp(f)} \leq \frac{4|I_i|}{\sp(f)}$ for $i=1,...,k$. Plugging the above back to Equation \ref{eqn: mgf bound} and using the fact that $\exp\bigg[\frac{r^2\sp(f)^2}{8l}\bigg] \geq 1$, we have
    \begin{align*}
        \bb E\bigg[\exp\bigg (r\frac{\sum_{i=1}^n G_i}{n}\bigg)\bigg] \leq (1 + \bar\beta/e^{2})\exp\bigg[\frac{r^2\sp(f)^2}{8l}\bigg].
    \end{align*}
    Applying Markov's inequality, we have for $\epsilon > 0$
    \begin{align}
    \label{eqn: markov bound before minimization}
        &\bb P\bigg[S_n - \bb E S_n \geq n\epsilon\bigg] = \bb P\bigg[\exp\frac{r}{n}(S_n - \bb E S_n) \geq e^{r\epsilon}\bigg]\nn\\
        &\leq \frac{\bb E [\exp\frac{r}{n}(S_n - \bb E S_n)]}{e^{r\epsilon}}\nn\\
        &\leq (1 + \bar\beta/e^{2})\exp\bigg[-r\epsilon + \frac{r^2\sp(f)^2}{8l}\bigg]. 
    \end{align}
    The right-hand side achieves minima w.r.t. $r$ when
    \begin{align*}
        r = \frac{4\epsilon l}{\sp(f)^2},
    \end{align*}
    which clearly satisfies $r \leq \frac{4l}{\sp(f)}$ when $\epsilon < \sp(f)$. Plugging the minimizer into Equation \eqref{eqn: markov bound before minimization} yields
    \begin{align}
        \bb P\bigg[S_n - \bb E S_n \geq n\epsilon\bigg] \leq (1 + \bar\beta/e^{2})\exp\bigg[-\frac{2l\epsilon^2}{\sp(f)^2}\bigg].\nn
    \end{align}
    Replacing $l$ by $\hat n = \lfloor n/k\rfloor = \lfloor n\lceil (10n/c)^{1/(\gamma+1)} \rceil^{-1}\rfloor$ gives the desired result 
    \begin{align}
        \bb P\bigg[S_n - \bb E S_n \geq n\epsilon\bigg] \leq (1 + \bar\beta/e^{2})\exp\bigg[-\frac{2\hat n\epsilon^2}{\sp(f)^2}\bigg],\nn
    \end{align}
    for $0<\epsilon<\sp(f)$.
\end{proof}

A similar proof gives an analogous Hoeffding-type inequality for exponential $\alpha$-mixing processes which is of independent interest. We document it here for completeness.

\begin{prop}
\label{thm: hoeffding for alpha mixing}
    Let $\bd{X}$ be a stationary $\alpha$-mixing sequence with the coefficient satisfying Definition \ref{def: exp mixing}. Assume that $f:\cal Y\to \bb R$ has bounded span, i.e., $\sp(f)< \infty$. Let $S_n = \sum_{i=0}^{n-1} f(X_i)$. Then, for all $\epsilon \in (0, \sp(f))$, it holds
    \begin{align}
    \label{eqn: hoeffding for alpha mixing}
        \bb P\bigg[S_n - \bb E S_n \geq n\epsilon\bigg] \leq (1 + 4e^{-2}\bar\alpha)\exp\bigg\{- \frac{2\hat n\epsilon^2}{\sp(f)^2} \bigg\},
    \end{align}
    where $\hat n= \lfloor n\lceil (10n/c)^{1/(\gamma+1)} \rceil^{-1}\rfloor$ and $c, \gamma$ are defined in Definition \ref{def: exp mixing}.
\end{prop}

\begin{proof}
    The proof is the same as that of Proposition \ref{thm: hoeffding for beta mixing} except for the part where $\Phi_2$ is estimated. In the $\alpha$-mixing case, $\Phi_2$ can be bounded by Lemma \ref{lem: alpha mixing ineq} and exponential $\alpha$-mixing condition (Definition \ref{def: exp mixing}). 
    \begin{align}
        \Phi_2 &= \bigg |\bb E\bigg[\prod_{j=1}^{|I_i|}\exp\bigg(r\frac{G_j}{|I_i|}\bigg)\bigg] -  \prod_{j=1}^{|I_i|} \bb E\bigg[\exp\bigg(r\frac{G_j}{|I_i|}\bigg)\bigg]\bigg |\nn\\
        &\stackrel{(a)}{\leq} 4\alpha(k)(|I_i|-1)\prod_{j=1}^{|I_i|}\bigg\|\exp\bigg[\frac{rG_j}{|I_j|}\bigg]\bigg\|\nn\\
        &\stackrel{}{\leq} 4\alpha(k)(|I_i|-1)e^{r\sp(f)}\nn\\
        &\stackrel{(b)}{\leq} e^{|I_i| -2} 4\bar\alpha(k) e^{-ck^\gamma} e^{r\sp(f)}\nn\\
        &\stackrel{}{\leq} 4e^{-2}\bar\alpha \exp\{|I_i| + r\sp(f) - ck^\gamma\},
    \end{align}
    where $(a)$ is due to Lemma \ref{lem: alpha mixing ineq} and $(b)$ is due to Definition \ref{def: exp mixing} and the fact $\|I_i\|-1\leq \exp(|I_i|-2)$. The rest of the proof follows that of Proposition \ref{thm: hoeffding for beta mixing}.
\end{proof}

\subsection{Proof of Theorem \ref{thm: mtbfa bound}}
\label{sec: mtbfa proof}

\subsubsection{Case 1: \texorpdfstring{$\beta$}{beta}-mixing}
\label{sec: case 1}
When $\bd X$ is exponential $\beta$-mixing satisfying Definition \ref{def: exp mixing}, we show the lower bound of $\mtbfa$ as follows. For exponential $\alpha$-mixing processes satisfying Definition \ref{def: exp mixing}, the proof follows the same procedure after replacing Proposition \ref{thm: hoeffding for beta mixing} with Proposition \ref{thm: hoeffding for alpha mixing}.

\begin{proof}
To determine the upper bound for $\mtbfa$, we condition on the fact that the change point $\tau$ is $\infty$. We use $\bb E_\infty$ and $\bb P_\infty$ to denote the expectation and the probability under $\tau=\infty$. For threshold $b > 0$, minimum burn-in period $M$, and stopping rule $T(b, M)$ in Equation \eqref{eqn: stopping rule}, the $\mtbfa$ reads
\begin{align}
\label{eqn: mtbfa bound 1}
    &\quad\bb E_\infty [T(b, M)] = \sum_{t=1}^\infty \bb P_\infty[T(b, M) \geq t] \nn\\
    &= M + \sum_{t=M+1}^\infty \bb P_\infty[T(b, M) \geq t]\nn \\
    &= M + \sum_{l=M+1}^\infty \left(1 - \mathbb{P}_\infty\left\{\bigcup_{t=M+1}^{l} \{T(b, M) = t\}\right\}\right) \nn\\
    &\stackrel{(a)}{\geq} M + \sum_{l=M+1}^\infty \left(1 - \mathbb{P}_\infty\left\{\bigcup_{t=M+1}^{l} \{\hat s_t > b\}\right\}\right)\nn\\
    &\geq M + \sum_{l=M+1}^L \left(1 - \mathbb{P}_\infty\left\{\bigcup_{t=M+1}^{l}\bigcup_{k=1}^{t-M}\{s_{k:t} \geq b\}\right\}\right)\nn\\
    &\stackrel{(b)}{\geq} M + \sum_{l=M+1}^L \left(1 - \sum_{t=M+1}^{l}\sum_{k=1}^{t-M}\mathbb{P}_\infty\{s_{k:t} \geq b\}\right)\nonumber\\
    &= M + \sum_{l=1}^{L-M}\left(1 - \sum_{t=1}^{l}\sum_{k=1}^{t}\mathbb{P}_\infty\{s_{k:t+M} \geq b\}\right),
\end{align}
where $(a)$ is due to the majorization of the event $\{T(b, M) = t\}$ to $ \{\hat s_t > b\}$, $(b)$ is due to the application of the union bound, and $M < L <\infty $ is an integer constant. 

To further lower-bound the right-hand side of the above, {we consider the tail probability in \eqref{eqn: mtbfa bound 1}. Due to stationarity, we can study $\bb P_\infty\{s_{1:t} \geq b\}$ for some $t\geq 1$ without loss of generality. Suppose we pick the offset parameter $\Delta = C(r, h) + \delta$ for some $\delta > 0$. By Lemma \ref{lem: empirical mmd consistency}, we known that $\bb E_\infty[s_{1:t}] =  t\bb E_\infty[s(\cal B_r(1)] \in [-tC_{\mu,\mu}(r, h) - t\delta, -t\delta]$ almost surely for sufficiently large $h$.} Additionally, one can verify that $\{\cal B_r(t)\}_{t=1}^{\infty}$ is $\beta$-mixing with coefficient $\tilde\beta(t) = \beta(tr)$. By assumption, $\beta(tr)$ satisfies Definition \ref{def: exp mixing} and 
\begin{align}
    \tilde\beta(t) = \beta(tr) \leq \bar\beta \exp(-cr^\gamma t^\gamma),\nn
\end{align}
for $\bar\beta, \gamma, c >0$. Thus, we can apply Proposition \ref{thm: hoeffding for beta mixing} to obtain a tail probability bound on $s_{1:t}$. For $t\geq 1$, 
\begin{align}
\label{eqn: tail bound before minimization}
    &\quad\bb P_\infty [s_{1:t} > b] \nn\\
    &= \mathbb{P}_\infty\left\{s_{1:t} - \bb E_\infty [s_{1:t}] > b - \bb E_\infty [s_{1:t}]\right\}\nn\\
    &\stackrel{(a)}{\leq} \frac{e^2+\bar\beta}{e^2}\exp\bigg\{- \frac{2\hat t(b - t\bb E_\infty[s(\cal B_r(1)])^2}{t^2\sp(\widehat{\mmd}[\cal B_r(1), \cal D_h])^2}\bigg\}\nn\\
    &\stackrel{(b)}{\leq} \frac{e^2+\bar\beta}{e^2}\exp\bigg[- \frac{\hat t(b + t\delta)^2}{t^2\bar{k}}\bigg]\nn\\
    &\stackrel{(c)}{\leq} \frac{e^2+\bar\beta}{e^2}\exp\bigg\{\left[1-\left(\frac{r^\gamma c}{10}\right)^{\frac{1}{1+\gamma}}\right]\frac{(\frac{b}{\sqrt{t}} + \sqrt{t}\delta)^2}{t^{\frac{1}{1+\gamma}}\bar{k}}\bigg\},
\end{align}
where in $(a)$ we apply Proposition \ref{thm: hoeffding for beta mixing} for the conditional probability under event $A$, in $(b)$ we apply $\sp(\widehat{\mmd}[\cal B_r(1), \cal D_h]) \leq \sqrt{2\bar{k}}$ which can deduced from \eqref{eqn: empirical mmd}, and in $(c)$ $\hat t$ is replace with its lower-bound,
\begin{align}
\label{eqn: shrink n beta}
    \hat t &= \bigg\lfloor \frac{t}{\lceil (10tr^{-\gamma}/c)^{\frac{1}{1+\gamma}} \rceil}\bigg\rfloor \geq \bigg\lfloor\frac{t}{(10tr^{-\gamma}/c)^{\frac{1}{1+\gamma}} + 1 }\bigg\rfloor\nn\\
    &\geq \bigg\lfloor\frac{t}{(10tr^{-\gamma}/c)^{\frac{1}{1+\gamma}} + t^{\frac{1}{1+\gamma}} }\bigg\rfloor\nn\\
    &\geq \bigg\lfloor \frac{t^{\frac{\gamma}{1+\gamma}}}{(10r^{-\gamma}/c)^{\frac{1}{1+\gamma}}}\bigg\rfloor \geq \frac{t^{\frac{\gamma}{1+\gamma}}}{(10r^{-\gamma}/c)^{\frac{1}{1+\gamma}}} - 1\nn\\
    &\geq [(r^\gamma c/10)^{\frac{1}{1+\gamma}}-1]t^{\frac{\gamma}{1+\gamma}},
\end{align}
assuming $c, r, \gamma$ are sufficiently large such that $(r^\gamma c/10)^{\frac{1}{1+\gamma}} > 1$. 

The right-hand side of \eqref{eqn: tail bound before minimization} achieves its maximum when $t = t^* \coloneqq \frac{(\gamma + 2)b}{\gamma\delta}$, which yields
\begin{align}
\label{eqn: tail bound after minimization}
    &\bb P_\infty [s_{1:t^*} > b]\leq D,
\end{align}
where 
\begin{align}
    D\coloneqq& \underbrace{(1 + \bar\beta e^{-2})\exp\bigg[\frac{\xi(\gamma, c, r)}{\bar{k}} b^{\frac{\gamma}{\gamma+1}} \delta^{\frac{\gamma + 2}{\gamma+1}}\bigg]}_{D_1}\nn\\
    \xi(\gamma, c, r) \coloneqq&  \left[1-\left(\frac{r^\gamma c}{10}\right)^{\frac{1}{1+\gamma}}\right]\bigg(\frac{2}{\gamma} + 2\bigg)^2\bigg(\frac{2}{\gamma}+1\bigg)^{-\frac{\gamma+2}{\gamma+1}}.\nn
\end{align}
Note that $\xi(\gamma, c, r)< 0$ for sufficiently large $c, r, \gamma$. 

Using \eqref{eqn: tail bound after minimization} and stationarity, each of $\mathbb{P}_\infty[s_{k:t+M} > b]$ in \eqref{eqn: mtbfa bound 1} can be upper-bounded by $\bb P_\infty [s_{1:t^*} > b]$, which yields 
\begin{align}
\label{eqn: mtbfa bound 2}
    &\bb E_\infty[T(b, M)] \geq M + \sum_{l=1}^{L-M}\left\{1 - \sum_{t=1}^{l}\sum_{k=1}^{t}\mathbb{P}_\infty[s_{1:t^*} > b]\right\}\nn\\
    &= L - D\sum_{l=1}^{L-M}l(l+1).
\end{align}
The right-hand side achieves maxima when $L = L^*$, where $L^*$ is obtained as the largest solution of 
\begin{align}
    (L^*-M)(L^*-M+1) = 2/D. \nn
\end{align}
Some simple calculation shows that 
\begin{align}
    L^* = M-\frac{1}{2} + \frac{1}{2}\sqrt{1 + 8/D}.\nn
\end{align}
Returning to \eqref{eqn: mtbfa bound 2}, we have
\begin{align}
\label{eqn: final ineq}
    &\bb E_\infty[T(b, M)]\nn\\
    &\geq L^* - D \sum_{l=1}^{L^*-M}l^2 + l\nn\\
    &\geq L^* - D (L^*-M)(L^*-M+1)(L^*-M+2)/6\nn\\
    &\geq M + \frac{1}{3}(L^* - M) - \frac{4}{3} \geq M -\frac{9}{6} +\frac{1}{6} \sqrt{1 + 8/D}\nn\\
    &\geq M + \frac{\sqrt{2}-9}{6} + \frac{2}{3} \sqrt{1/D}\nn,
\end{align}
where the second to last inequality uses the fact that $\sqrt{a+b}\geq\sqrt{a/2}+\sqrt{b/2}$ for $a, b\geq 0$. Plugging the value of $D$ yields the lower bound in $\eqref{eqn: mtbfa beta}$.
\end{proof}

\subsubsection{Case 2: \texorpdfstring{$\phi$}{phi}-mixing}
\label{sec: case 2}
When $\bd X$ is $\phi$-mixing satisfying Definition \ref{def: fast phi mixing}, we show the lower bound of $\mtbfa$ as follows. For exponential $\alpha$-mixing processes satisfying Definition \ref{def: exp mixing}, the proof follows the same procedure after replacing Proposition \ref{thm: hoeffding for beta mixing} with Proposition \ref{thm: hoeffding for alpha mixing}.

\begin{proof}
    Follow the same argument in Case \ref{sec: case 1} until the application of Proposition \ref{thm: hoeffding for beta mixing}. Note that $\{\cal B_r(t)\}_{t=1}^\infty$ is $\phi$-mixing with coefficient $\tilde\phi(t) = \phi(tr)$. Thus, $\{\cal B_r(t)\}_{t=1}^\infty$ satisfies Definition \ref{def: fast phi mixing} with constant $\Phi$ as soon as $\bd X$ does. Then, we can apply Proposition \ref{thm: hoeffding for phi mixing} to obtain a tail bound on $s_{1:t}$. For $t\geq 1$,
    \begin{align}
        &\bb P_\infty [s_{1:t} > b] = \mathbb{P}_\infty\left\{s_{1:t} - \bb E_\infty [s_{1:t}] > b - \bb E_\infty [s_{1:t}]\right\} \nn\\
        &\leq \exp\bigg\{- \frac{2(b - t\bb E_\infty[s(\cal B_r(1))])^2}{t(2\Phi+1)^2\sp(\widehat{\mmd}[\cal B_r(1), \cal D_h])^2}\bigg\}\nn\\
        &\leq \exp\bigg\{- \frac{(b + t\delta)^2}{t\bar{k}}\bigg\}.\nn
    \end{align}
    The right-hand side of the above achieves its maximum when $t = t^*\coloneqq b/\delta$, which yields
    \begin{align}
        \bb P_\infty [s_{1:t^*} > b] \leq \exp\bigg[-\frac{4b\delta}{\bar{k}}\bigg],
    \end{align}
    The rest of the proof follows the same procedure in Case \ref{sec: case 1}. Then, we have 
    \begin{align}
        &\bb E_\infty[T(b, M)] \nn\\
        &\geq M + \frac{\sqrt{2}-9}{6} + \frac{2}{3}\exp\bigg[\frac{2b \delta}{\bar k}\bigg].\nn
    \end{align}
    We arrive at the $\mtbfa$ lower-bound in \eqref{eqn: mtbfa phi}.
\end{proof}

\subsection{Proof of Theorem \ref{thm: md bound}}
\label{sec: md proof}

\subsubsection*{Case 1: \texorpdfstring{$\beta$}{beta}-mixing}
\label{sec: md case 1}
When $\bd X$ is exponential $\beta$-mixing satisfying Definition \ref{def: exp mixing}, we show the upper bound of $\md$ as follows. Note the offset parameter $\Delta = C_{\mu, \mu}(r,h) + \delta$ as defined in the proof of Theorem \ref{thm: mtbfa bound} in Appendix \ref{sec: mtbfa proof}.

\begin{proof}
    Assuming the change point $\tau=0$, we denote the probability and expectation under the alternative as $\bb P_0$ and $\bb E_0$, respectively. For $t\geq 1$, $\bb E_0 \{\widehat{\mmd}[\cal B_r(t), \cal D_h]\} \geq \mmd_k(\mu, \nu)-C(r,h) - \delta$ almost surely for sufficiently large $h$ according to Lemma \ref{lem: empirical mmd consistency}. Let $D(\mu, \nu):=\mmd_k(\mu, \nu)- C(r,h)$. Now, we can write the average detection delay as follows
    \begin{align}
    \label{eqn: md bound 1}
        \bb E_0 &[T(b, M)] = \sum_{t=1}^\infty \bb P_0[\hat{s}_t \leq b] \leq \sum_{t=1}^\infty \bb P_0[s_{1: t} \leq b]\nn\\
        &\leq \sum_{t=1}^{t_0}\bb P_0[s_{1: t} \leq b] + \sum_{t=t_0+1}^\infty \bb P_0[s_{1: t} \leq b]\nn\\
        &\leq \max\bigg\{M, \frac{b}{D(\mu, \nu) - \Delta - \delta}\bigg\} + \sum_{t=t_0+1}^\infty \bb P_0[s_{1: t} < b],
    \end{align}
    where $t_0 =\max\left\{M,\frac{b}{D(\mu, \nu)- \Delta - \delta}\right\}$. In the rest of the proof, we aim to show that the second term on the right-hand side is ultimately bounded by a constant multiple of the first term, and the desired result is reached.
    
    To bound the second term, we apply Proposition \ref{thm: hoeffding for beta mixing} to get the tail probability bound of $s_{1:t}$, similarly to the proof in Appendix \ref{sec: mtbfa proof}. For $t \geq t_0 + 1$, we have
    \begin{align*}
        \mathbb{P}_0 & [s_{1:t}\leq b] \leq \mathbb{P}_0\left\{s_{1:t} - \bb E_0 [s_{1:t}]\leq  b - \bb E_0 [s_{1:t}]\right\}\nonumber\\
        \stackrel{(a)}{\leq}& (1 + \bar\beta/e^{2})\exp\left\{-\frac{2\hat t(\bb E_0 [s_{1:t}]-b)^2}{t^2\sp(\widehat{\mmd}[\cal B_r(t), \cal D_h])^2}\right\}\\
        \stackrel{(b)}{\leq}& (1 + \bar\beta/e^{2})\exp\left\{-\frac{\hat t[t(D(\mu, \nu) - \Delta-\delta)-b]^2}{t^2\bar k}\right\}\\
        \stackrel{(c)}{\leq}& (1 + \bar\beta/e^{2})\exp\left\{-\psi(c, \gamma) \frac{[t(D(\mu, \nu) - \Delta-\delta)-b]^2}{t^{\frac{2+\gamma}{1+\gamma}}\bar k}\right\},
    \end{align*}
    where $\sp(\widehat{\mmd}[\cal B_r(t), \cal D_h]) \leq \sqrt{2\bar k}$ can be deduced from Equation \eqref{eqn: empirical mmd}, $(a)$ follows from Proposition \ref{thm: hoeffding for beta mixing} and $\hat t = \lfloor t\lceil (10t/c)^{1/(\gamma+1)} \rceil^{-1}\rfloor$, $(b)$ uses stationarity of the post-change process and the conditioning on $A'$, and $(c)$ follows from the relation in Equation \eqref{eqn: shrink n beta} and $\psi(c, \gamma) \coloneqq (10/c + 1)^{-\frac{1}{1+\gamma}}-1$. After magnifying the exponential term by setting $b$ to $t_0(D(\mu, \nu) - \Delta-\delta)$ and splitting the summation at $\bar t \coloneqq \lceil t_0^\sigma\rceil$ for some $\sigma\in (\frac{2+\gamma}{2(1+\gamma)}, 1)$, the second term on the right-hand side of Equation \eqref{eqn: md bound 1} becomes
    \begin{align}
         &\frac{1}{1 + \bar\beta/e^{2}}\sum_{t=t_0+1}^\infty \bb P_0 \nn\\ 
         &\leq \sum_{t=t_0+1}^\infty \exp\left\{-\psi(c, \gamma) \frac{(D(\mu, \nu) - \Delta-\delta)^2(t-t_0)^2}{t^{\frac{2+\gamma}{1+\gamma}}\bar k}\right\}\nn\\
         &= \sum_{t=1}^{\bar t-1} \exp\left\{-\psi(c, \gamma) \frac{(D(\mu, \nu) - \Delta-\delta)^2t^2}{(t+t_0)^{\frac{2+\gamma}{1+\gamma}}\bar k}\right\}\nn\\
         &+ \sum_{t=\bar t}^{\infty} \exp\left\{-\psi(c, \gamma) \frac{(D(\mu, \nu) - \Delta-\delta)^2t^2}{(t+t_0)^{\frac{2+\gamma}{1+\gamma}}\bar k}\right\}\nn\\
         & \leq (\bar t - 1) \exp\left\{-\psi(c, \gamma) \frac{(D(\mu, \nu) - \Delta-\delta)^2}{(1+t_0)^{\frac{2+\gamma}{1+\gamma}}\bar k}\right\}\nn\\
         &+ \sum_{t=\bar t}^{\infty} \exp\left\{-\psi(c, \gamma) \frac{(D(\mu, \nu) - \Delta-\delta)^2t^{\frac{\gamma}{1+\gamma}}}{(1+t^{\frac{1}{\sigma} - 1})^{\frac{2+\gamma}{1+\gamma}}\bar k}\right\},\nn
    \end{align}
    where the first term is magnified by fixing $t=1$ for all summands, and the second term is magnified by majorizing the denominator inside the exponential for each summand via $(1+t_0/t)^\frac{2+\gamma}{1+\gamma}\bar k \leq (1+\bar t^\frac{1}{\sigma}/t)^\frac{2+\gamma}{1+\gamma}\bar k \leq (1+t^{\frac{1}{\sigma} - 1})^{\frac{2+\gamma}{1+\gamma}}\bar k$.
    
    At this point, we can compare the growth rate of the two terms above and $t_0$, which yields
    \begin{align*}
        &\lim_{t_0\to\infty} \frac{(\bar t -1)}{t_0} \exp\left\{-\psi(c, \gamma) \frac{(D(\mu, \nu) - \Delta-\delta)^2}{(1+t_0)^{\frac{2+\gamma}{1+\gamma}}\bar k}\right\} = 0,\\
        &\lim_{t_0\to\infty} \sum_{t=\bar t}^{\infty} \exp\left\{-\psi(c, \gamma) \frac{(D(\mu, \nu) - \Delta-\delta)^2t^{\frac{\gamma}{1+\gamma}}}{(1+t^{\frac{1}{\sigma} - 1})^{\frac{2+\gamma}{1+\gamma}}\bar k}\right\} = 0,
    \end{align*}
    where both equations above follow from $\sigma\in (\frac{2+\gamma}{2(1+\gamma)}, 1)$. Thus, we have
    \begin{align}
     \label{eqn: md bound 2}
        \lim_{t_0\to\infty} \frac{1}{t_0}\sum_{t=t_0+1}^\infty \bb P_0[s_{1: t} \leq b] = 0.
    \end{align}
    We have reached the desired result in Equation \eqref{eqn: md bound} after combining Equation \eqref{eqn: md bound 1} and \eqref{eqn: md bound 2}.
\end{proof}

\subsubsection*{Case 2: \texorpdfstring{$\phi$}{phi}-mixing}
\label{sec: md case 2}
When $\bd X$ is $\phi$-mixing satisfying Definition \ref{def: fast phi mixing}, the upper bound of $\md$ is shown to follow Equation \ref{eqn: md bound} using the same recipe as in Appendix \ref{sec: md case 1}. Using Proposition \ref{thm: hoeffding for phi mixing}, the second term on the right-hand side of Equation \ref{eqn: md bound 1} can also be proven ultimately negligible compared to the first term. 

\begin{proof}
    We shall directly start bounding the second term on the right-hand side of Equation \ref{eqn: md bound 1} using Proposition \ref{thm: hoeffding for phi mixing}, which is written as
    \begin{align*}
        &\sum_{t=t_0+1}^\infty\mathbb{P}_0[s_{1:t}\leq b] \\
        =& \sum_{t=t_0+1}^\infty \mathbb{P}_0\left\{s_{1:t} - \bb E_0 [s_{1:t}]\leq b - \bb E_0 [s_{1:t}] | A'\right\}\nonumber\\
        \stackrel{}{\leq}& \sum_{t=t_0+1}^\infty \exp\left\{-\frac{2(\bb E_0 [s_{0:t}]-b)^2}{t(2\Phi +1)^2\sp(\widehat{\mmd}[\cal B_r(t), \cal D_h])^2}\right\}\\
        \stackrel{}{\leq}& \sum_{t=t_0+1}^\infty \exp\left\{-\frac{[t(D(\mu, \nu) - \Delta-\delta)-b]^2}{t\bar k}\right\}.
    \end{align*}
    If we magnify the exponential term by setting $b$ to $t_0(D(\mu, \nu) - \Delta-\delta)$ and split the summation at $\bar t \coloneqq \lceil t^{2/3}_0 \rceil$, then it becomes
    \begin{align*}
        &\sum_{t=t_0+1}^\infty \bigg\{\mathbb{P}_0[s_{1:t}\leq b]\bigg\} \\
        &\leq \sum_{t=t_0+1}^\infty \exp\left\{-\frac{(D(\mu, \nu) - \Delta-\delta)^2(t-t_0)^2}{t\bar k}\right\}\\
        &= \sum_{t=1}^{\bar t-1} \exp\left\{-\frac{(D(\mu, \nu) - \Delta-\delta)^2t^2}{(t+t_0)\bar k}\right\}\\
        &+ \sum_{t=\bar t}^{\infty} \exp\left\{- \frac{(D(\mu, \nu) - \Delta-\delta)^2t^2}{(t+t_0)\bar k}\right\}\nn\\
        &\leq (\bar t - 1)\exp\left\{-\frac{(D(\mu, \nu) - \Delta-\delta)^2}{(1+t_0)\bar k}\right\} \nn\\
        &+ \sum_{t=\bar t}^{\infty} \exp\left\{- \frac{(D(\mu, \nu) - \Delta-\delta)^2t}{(1+t^{1/2})\bar k}\right\}.
    \end{align*}
    At this point, we can easily verify that both terms on the right-hand side are ultimately negligible compared to $t_0$, and the proof is complete. 
\end{proof}

\subsection{MMD-CUSUM Test Pseudocode}
\SetKwComment{Comment}{/* }{ */}
\RestyleAlgo{ruled}
\SetAlgoNoLine
\LinesNotNumbered
\begin{algorithm}[h]
\caption{MMD-CUSUM test}\label{alg: MMD CUSUM}
\KwData{Data stream $\{X_i\}$, reference data $\cal D_h$ of size $h$; empty buffer $\cal B_r$ of size $r$; $s_{\min}$ stores the min of the partial sum; pick $\Delta$ during calibration; pick threshold $b>0$.}
$i, t\gets 0$\;
$\hat s_0, s_{0}, s_{\min} \gets 0$\;
$\cal B_r \gets \varnothing$\;
\While{$\hat s_{t} \leq b$}{
    $\cal B_r\gets \cal B_r\cup X_i$\;
  \If{$(i \mod r) = 0$}{
    $s_{t+1}\gets s_{t} + \mmd[\hat\mu_r, \hat\nu_h] - \Delta$\;
    $s_{\min} \gets \min\{s_{t+1}, s_{\min}\}$\;
    $\hat s_{t+1} \gets s_{t+1} - s_{\min}$\;
    $\cal B_r \gets \varnothing$\;
    $t\gets t+1$\;
  }
  $i \gets i + 1$\;
}
\end{algorithm}

\section*{ACKNOWLEDGMENT}
The authors gratefully acknowledge the insightful discussions with our collaborators from Cisco Systems on potential applications of this work, including Ashish Kundu, Jayanth Srinivasa, and Hugo Latapie. Hao Chen and Abhishek Gupta's work was supported by Cisco Systems grant GR127553. Yin Sun's work was supported in part by the NSF under grant No. CNS-2239677, and by the ARO under grant No. W911NF-21-1-0244. Ness Shroff's work has been supported in part by NSF grants: CNS-2312836, CNS-2223452, CNS-2225561, CNS-2112471, CNS-2106933, a grant from the Army Research Office: W911NF-21-1-0244, and was sponsored by the Army Research Laboratory under Cooperative Agreement Number W911NF-23-2-0225. The views and conclusions contained in this document are those of the authors and should not be interpreted as representing the official policies, either expressed or implied, of the Army Research Laboratory or the U.S. Government. The U.S. Government is authorized to reproduce and distribute reprints for Government purposes, notwithstanding any copyright notation herein.

\bibliographystyle{ieeetr}
\bibliography{ref.bib}

\end{document}